\def\version{USENIX}
\newif \ifsubmission \submissionfalse

\newif \iffull 
\newif \ifACM
\newif \ifUSENIX
\newif \ifIEEE
\newif \ifLNCS

\newif \ifCCS
\newif \ifSP
\newif \ifNDSS
\newif \ifCrypto
\newif \ifFC

\def\fullstring{full}
\def\ACMstring{ACM}
\def\USENIXstring{USENIX}
\def\IEEEstring{IEEE}
\def\LNCSstring{LNCS}

\def\CCSstring{CCS}
\def\SPstring{SP}
\def\NDSSstring{NDSS}
\def\Cryptostring{CRYPTO}
\def\FCstring{FC}

\ifx \version\fullstring \fulltrue \fi
\ifx \version\ACMstring \ACMtrue \fi
\ifx \version\USENIXstring \USENIXtrue \fi
\ifx \version\IEEEstring \IEEEtrue \fi
\ifx \version\LNCSstring \LNCStrue \fi

\ifx \version\CCSstring \ACMtrue \CCStrue \fi
\ifx \version\SPstring \IEEEtrue \SPtrue \fi
\ifx \version\NDSSstring \IEEEtrue \NDSStrue \fi
\ifx \version\Cryptostring \LNCStrue \Cryptotrue \fi
\ifx \version\FCstring \LNCStrue \FCtrue \fi

\iffull \documentclass[11pt]{article}
\usepackage[utf8]{inputenc}
\usepackage[letterpaper,margin=1in,bmargin=1.5in]{geometry}

\usepackage[pdfauthor={},pdfpagelabels=true,linktocpage=true,hidelinks,linktoc=all,colorlinks]{hyperref}
\hypersetup{linkcolor = {teal}, citecolor = {magenta}, urlcolor = {black}}

\usepackage[style=numeric-comp,backend=bibtex,maxbibnames=4]{plain}

\date{}
 \fi
\ifACM 
\PassOptionsToPackage{dvipsnames}{xcolor}
\documentclass[sigconf]{acmart}

\setcopyright{none}
\renewcommand\footnotetextcopyrightpermission[1]{} \fi
\ifUSENIX 

\PassOptionsToPackage{table,dvipsnames}{xcolor}
\documentclass[letterpaper, twocolumn, 10pt]{article}
\usepackage{Conferences/USENIX/usenix} \fi
\ifIEEE 
\documentclass[letterpaper,conference,compsoc]{IEEEtran}
\usepackage[noadjust]{cite}
\usepackage[scaled=0.88]{helvet}

 \fi
\ifLNCS 
\PassOptionsToPackage{dvipsnames}{xcolor}
\documentclass[runningheads]{Conferences/LNCS/llncs}

\ifsubmission \author{} \institute{} \fi
 \fi

\iffull \bibliography{references} \fi
\newif \ifcomments \commentsfalse
\newif \ifanon \anonfalse

\ifUSENIX \else \usepackage[table]{xcolor} \fi
\usepackage{xurl}
\usepackage{xspace}
\usepackage{hyperref}
\usepackage{amsmath}
\hypersetup{
    colorlinks=true,
    allcolors=blue,
    linkcolor=red,
    filecolor=magenta,
}
\usepackage[many]{tcolorbox}        
\usepackage{lipsum}
\usepackage{cleveref}
\usepackage{enumitem}
\usepackage{booktabs}
\usepackage{enumitem}
\usepackage{makecell}
\usepackage{multirow}
\usepackage[figuresright]{rotating}
\usepackage{array}
\usepackage{nicematrix}
\usepackage{tikz}
\usepackage{subcaption}
\usepackage{booktabs}
\usepackage{siunitx}
\usepackage{algorithm}
\usepackage{algpseudocode}
\usepackage[cal=cm]{mathalfa}
\usepackage{mdframed}
\usepackage{comment}
\usepackage{enumitem}
\usepackage{tcolorbox}

\setlist[itemize]{itemsep=0.2em, parsep=0pt, topsep=0.3em}

\usepackage{titlesec}
\titleformat{\paragraph}[runin]{\normalfont\bfseries}{\theparagraph}{0.75em}{}
\titlespacing*{\paragraph}{0pt}{1ex}{0.75em} 

\algtext*{EndFor}
\algtext*{EndIf}

\usetikzlibrary{shapes, arrows.meta, positioning, calc, math, tikzmark, fit}

\usepackage{amsfonts}
\ifLNCS\else\usepackage{amsthm}\fi
\ifACM\else\usepackage{amssymb}\fi

\newcommand{\yes}{\textsf{yes}}
\newcommand{\no}{\textsf{no}}

\definecolor{ForestGreen}{RGB}{34,139,34}

\newcommand{\talprtb}{\ensuremath{d}\xspace}

\ifcomments
    \newcommand{\mahimna}[1]{\textsf{\small{\color{violet!80}{[Mahimna: {#1}]}}}}
    \newcommand{\kushal}[1]{\textsf{\small{\color{blue}{[Kushal: {#1}]}}}}
    \newcommand{\james}[1]{\textsf{\small{\color{green!75!black}{[James: {#1}]}}}}
    \newcommand{\ari}[1]{\textsf{\small{\color{red}{[Ari: {#1}]}}}}
    \newcommand{\jay}[1]{\textsf{\small{\color{orange}{[Jay: {#1}]}}}}
    \newcommand{\sarah}[1]{\textsf{\small{\color{red}{[Sarah: {#1}]}}}}
    \newcommand{\andres}[1]{\textsf{\small{\color{blue}{[Andres: {#1}]}}}}
    
    \newcommand{\dani}[1]{\textsf{\small{\color{purple}{[Dani: {#1}]}}}}
    \newcommand{\sam}[1]{\textsf{\small{\color{yellow!75!black}{[Sam: {#1}]}}}}
    \newcommand{\paddy}[1]{\textsf{\small{\color{pink}{[Paddy: {#1}]}}}}
\else
    \newcommand{\kushal}[1]{}
    \newcommand{\mahimna}[1]{}
    \newcommand{\james}[1]{}
    \newcommand{\ari}[1]{}
    \newcommand{\jay}[1]{}
    \newcommand{\sarah}[1]{}
    \newcommand{\andres}[1]{}
    \newcommand{\dani}[1]{}
    \newcommand{\sam}[1]{}
    \newcommand{\paddy}[1]{}
\fi

\newcommand{\disclosure}{\mathsf{tally}}
\newcommand{\bm}{\alpha}

\DeclareSymbolFont{CMletters}{OML}{cmm}{m}{it}
\DeclareMathSymbol{\cmepsilon}{\mathord}{CMletters}{"0F}

\newtheorem{example}{Example}

\tcbset{
    sharp corners,
    colback = white,
    before skip = 0.2cm,    
    after skip = 0.5cm      
}                           

\newtcolorbox{boxA}{
    fontupper = \bf,
    boxrule = 1.5pt,
    colframe = black 
}

\iffull 
\else  \fi

\iffull 
\else  \fi

\ifLNCS\else

\iffull \newtheorem{theorem}{Theorem}[section]
\else \newtheorem{theorem}{Theorem} \fi
\newtheorem{corollary}{Corollary}[theorem]
\newtheorem{lemma}[theorem]{Lemma}

\newtheorem{definition}[theorem]{Definition}

\theoremstyle{remark}

\fi

\definecolor{keyFindingColor}{HTML}{4A90E2}   
\definecolor{keyFindingBackground}{HTML}{e9edf5} 

\usepackage[many]{tcolorbox}
\usepackage{enumitem}

\newlist{keyfindings}{itemize}{1}
\setlist[keyfindings,1]{%
  label=\(\triangleright\), 
  leftmargin=1.4em,
  itemsep=2pt plus 1pt minus 1pt,
  topsep=4pt plus 1pt minus 1pt
}

\newtcolorbox{keyFinding}[1][]{%
  enhanced, breakable,
  colback=keyFindingBackground,
  colframe=keyFindingColor,
  coltitle=black,
  boxrule=0pt, frame hidden,
  borderline west={3pt}{0pt}{keyFindingColor},
  sharp corners,
  left=10pt,right=10pt,top=8pt,bottom=10pt,
  before skip=8pt plus 2pt minus 2pt,
  after  skip=10pt plus 2pt minus 2pt,
  fonttitle=\bfseries,
  title=Key Finding, 
  attach boxed title to top left={yshift=-2mm, xshift=8pt},
  boxed title style={colback=white,colframe=keyFindingColor,boxrule=0.5pt,sharp corners,
                     left=8pt,right=8pt,top=2pt,bottom=2pt},
  before upper=\vspace{4pt},
  #1 
}

\begin{document}

\title{B-Privacy: Defining and Enforcing Privacy in Weighted Voting}

\ifSP
\IEEEoverridecommandlockouts
\makeatletter
\newcommand{\linebreakand}{%
  \end{@IEEEauthorhalign}
  \hfill\mbox{}\par
  \mbox{}\hfill\begin{@IEEEauthorhalign}
}
\makeatother
\fi
\ifanon\else
\author{
{\rm Samuel Breckenridge$^*$}\\
Cornell Tech, IC3
\and
{\rm Dani Vilardell$^*$}\\
Cornell Tech, IC3
\and
{\rm Andr\'{e}s F\'{a}brega}\\
Cornell Tech, IC3
\and
{\rm Amy Zhao}\\
Ava Labs, IC3
\and
{\rm Patrick McCorry}\\
Arbitrum Foundation
\and
{\rm Rafael Solari}\\
Tally
\and
{\rm Ari Juels}\\
Cornell Tech, IC3
}
\fi

\ifACM \renewenvironment{abstract}
  {\par\noindent\textit{\textbf{Abstract—}}\ignorespaces}
  {\par}
  





\begin{abstract}
\bfseries
In traditional, one-vote-per-person voting systems, \textit{privacy} equates with ballot secrecy: voting tallies are published, but individual voters' choices are concealed. 

Voting systems that weight votes in proportion to token holdings, though, are now prevalent in cryptocurrency and web3 systems. We show that these \textit{weighted-voting} systems overturn existing notions of voter privacy. Our experiments demonstrate that even with secret ballots, publishing raw tallies often reveals voters' choices.

Weighted voting thus requires \textit{a new framework for privacy}. We introduce a notion called \textit{B-privacy} whose basis is \textit{bribery}, a key problem in voting systems today. B-privacy captures the \textit{economic cost} to an adversary of bribing voters based on revealed voting tallies.

We propose a mechanism to boost B-privacy by noising voting tallies. We prove bounds on its tradeoff between B-privacy and transparency, meaning reported-tally accuracy. Analyzing 3,582 proposals across 30 Decentralized Autonomous Organizations (DAOs), we find that \textit{the prevalence of large voters (``whales'') limits the effectiveness of any B-Privacy-enhancing technique}. However, our mechanism proves to be effective in cases without extreme voting weight concentration: among proposals requiring coalitions of $\geq5$ voters to flip outcomes, our mechanism raises B-privacy by a geometric mean factor of $4.1\times$.

Our work offers the first principled guidance on transparency-privacy tradeoffs in weighted-voting systems, complementing existing approaches that focus on ballot secrecy and revealing fundamental constraints that voting weight concentration imposes on privacy mechanisms.
\end{abstract}

 \maketitle 
\else \maketitle  \fi

\def\thefootnote{*}\footnotetext{These authors contributed equally to this work.}
\pagestyle{plain}


\section{Introduction}

In traditional one-person-one-vote systems, privacy equates with ballot secrecy: aggregate tallies are published but individual choices remain hidden~\cite{aidt2017open, rhim2012ballotssecretfutilitysocial, gerber2013perceptions, seidmann_theory_2011}.

In this work, we show that \textit{weighted voting} fundamentally alters this picture of privacy in voting systems.
Long used for share-based voting in corporate governance \cite{lee2024shareholder}, 
weighted voting has also become the predominant tool for governance in blockchain protocols and DAOs~\cite{buterin2014daos} as well as for delegated voting in proof-of-stake consensus \cite{king2012ppcoin}.

Unlike traditional voting systems that assign equal weight to each vote, weighted systems allocate influence proportionally to token (or share) ownership. 
We show that this proportionality creates new privacy risks: even when individual ballots are hidden, the tallies themselves often reveal how participants voted. A simple example illustrates the problem. 
\begin{example}[Tally leakage]\label{prec_attack_ex}
   Alice possesses $1.2$ voting weight, while all other voters possess exactly $1$ voting weight. In a weighted tally---e.g., weight 1507 voting $\textnormal{\textsf{yes}}$ vs. weight 2510.2 voting $\textnormal{\textsf{no}}$---if one choice has a $.2$ fractional part, \textbf{it must correspond to Alice's choice}.
\end{example}

This is a toy example. In this work, however, we conduct experiments on 3,844 recorded proposal votes across 31 DAOs. (DAOs today lack persistent ballot secrecy, but are advancing toward it~\cite{aztec2023nounsdao,shutter2023shieldedvoting}.) We show that, even with hypothetical ballot secrecy, weighted tallies in the DAOs under study \textit{reveal a significant fraction of individual voters' choices}. 

\paragraph{Addressing weighted-voting privacy.} Weighted-voting systems thus require a \textit{new privacy framework}, which we introduce in this work.

Strong privacy is achievable in the weighted setting by suppressing tally details and publishing only the winner. Such redaction, however, \textit{would sacrifice transparency}, omitting critical statistics---such as the margin of victory and the rate of voter participation---that are non-negotiable for achieving trust in governance. 

Two key questions thus arise: 

\begin{enumerate}
    \item[Q1.] \textbf{How can we effectively \textit{measure} the privacy associated with published tallies in weighted voting?}
    \item[Q2.] \textbf{How can we \textit{enforce} privacy while preserving transparency for published tallies in weighted voting?}
\end{enumerate}

Answering these questions requires a privacy framework tailored to weighted voting, since existing notions like ballot secrecy were designed for one-person-one-vote systems and fail to capture the risks posed by tallies themselves.

Key among these risks today is the rampant, growing problem of \textit{bribery} / \textit{vote-buying}. For example, vote-buying constitutes a \$250+ million market in the Curve protocol~\cite{lloyd2023emergent}, while the LobbyFi vote-buying protocol commands 8–14\% of votes on major proposals in the popular Arbitrum L2 blockchain.


Tally privacy in the weighted-voting setting connects directly to bribery. For bribery strategy to be effective, an adversary must condition payouts on voter choices in a way that rewards compliance. To do so, the adversary must be able to deduce---or at least accurately estimate---how voters cast their votes. Adversarial exploitation of information in published tallies thus makes bribery strategies enforceable.


This observation motivates the new privacy notion we introduce in this work: \textbf{B-privacy} (short for ``Bribery-privacy''). B-privacy measures the \textit{economic cost} to an adversary of bribing voters given the information revealed in a published weighted-vote tally. B-privacy is complementary to standard coercion-resistance: while coercion-resistance prevents voters from cooperating with adversaries to prove how their ballot was cast, B-privacy addresses what the adversary can infer with tally access alone. Together, these properties provide comprehensive protection against bribery attacks. B-privacy thus addresses Q1 above with a concrete, economically grounded measure of privacy-related risk.

B-privacy also offers a foundation for broader reasoning about privacy. We prove in this work that at the level of individual voters, susceptibility to bribery relates to the common privacy concept of \textit{plausible deniability}---whether or not a voter's choice can be deduced from published vote data.

\paragraph{B-Privacy.} Informally, the B-privacy of a system is the \textit{minimum bribe} an adversary must pay to voters to achieve a desired outcome with a certain probability $p$ (e.g., to ensure a \textsf{yes} outcome in a \textsf{yes}/\textsf{no} vote). 

B-privacy is measured for a particular \textit{tally algorithm}, which specifies how tallies are published. For example, a tally algorithm might publish individual cleartext ballots (resulting in minimal B-privacy, as an adversary can pay out perfectly targeted bribes). Or it might publish only the winner (maximizing B-privacy, but eroding transparency). 

We define B-privacy in terms of a \emph{bribery game}, a Bayesian game between the adversary and a group of rational voters. In this game, the adversary specifies bribe amounts and conditions of payment based on the published tally. (E.g., bribes might be paid if the \textsf{yes} / \textsf{no} winning margin exceeds 10\%.) Voters then vote in a way that maximizes their expected utility, which combines their private utilities for vote outcomes and the potential bribe. B-privacy is defined as the minimum cost for the adversary to achieve its desired outcome in this game with a given probability $p$ at \emph{Bayesian Nash equilibrium}.

B-privacy is grounded in strategic behavior and cost, rather than idealized notions of secrecy, and so offers a practical lens on weighted-voting privacy. While no system can eliminate bribery, we introduce a tally algorithm that can boost B-privacy while retaining strong transparency.

\paragraph{Enforcing B-privacy via noising.} We introduce a simple tally algorithm that \textit{adds (Laplacian) noise} to a published tally. In answer to Q2 above, we show that this algorithm boosts B-privacy---i.e., yields a higher cost of bribery---while minimally perturbing tallies. 

Our approach recalls techniques for enforcing differential privacy. A subtle but critical issue, however, is that adding noise can \textit{flip a proposal's reported outcome}, implying an erroneous outcome. For example, in a 49.9\% \textsf{yes} vs. 50.1\% \textsf{no} vote, adding 0.2\% noise in the \textsf{yes} direction would flip the reported outcome from \textsf{no} to \textsf{yes}.

Our tally algorithm thus also \textit{corrects} the published outcome if necessary. This requirement, however, causes classical differential privacy bounds to break down. We therefore introduce \textit{new proof techniques} to obtain bounds on the B-privacy of our tally algorithm---one of our key contributions.

\paragraph{Limitation:} In this initial exploration of B-privacy, we restrict our focus to \textit{binary} voting choices, notionally \textsf{yes}/\textsf{no}. We do not consider multi-choice ballots or the impact of abstention (which may be treated as a ballot option). Extension to multiple choices renders analysis much more complex. We conjecture that our results still hold directionally for the multi-choice case, a question that we leave for future work. 

\subsection*{Contributions}
Ours is the first framework that quantifies bribery risk economically and provides a tunable mechanism that balances privacy and transparency in practice for weighted voting. We review related work in~\Cref{sec:related} and give preliminary formalism in~\Cref{sec:prelims}. Our contributions are then as follows: 

\begin{itemize}
    \item \textbf{Initiating study of weighted-voting privacy:} We introduce the problem of tally privacy for weighted-voting. We demonstrate the urgency of the problem with experimental, practical attacks on real-world DAO voting that reveal voter choices despite ballot secrecy (\Cref{sec:pd_in_daos}). 
    \item \textbf{B-privacy:} We introduce and formally define B-privacy, an economic measure of bribery resistance that is the first privacy measure specific to weighted-voting systems (\Cref{sec:b-privacy}). We prove basic results on B-privacy and explore methods for computing it in practice (\Cref{sec:B-privacy_comp}).
    \item \textbf{Noise-based privacy mechanism:} We present a simple, tunable noising mechanism that preserves winner correctness. We prove bounds on its balance between B-privacy and transparency (\Cref{sec:noise_mechanism}). 
    \item \textbf{Empirical analysis:} We analyze the effect of our noising mechanism across 3,582 proposals in 30 DAOs, revealing that extreme voting weight concentration fundamentally limits B-privacy improvements in most real-world cases. However, among proposals without whale dominance (requiring coalitions of $\geq5$ voters to flip the outcome), our mechanism improves B-Privacy by a geometric mean factor of $4.1\times$ over raw tallies, with minimal transparency degradation (\Cref{sec:noise_experiments}).

\end{itemize}

We conclude in~\Cref{sec:conclusion}. We relegate to the paper appendices theorems and proofs (\Cref{app:optimal_strategy}) and details on computational methods for B-privacy (\Cref{app:comp-methods}).

\section{Related Work}
\label{sec:related}

\paragraph{Voting privacy fundamentals.} 
Traditional privacy concepts for equal-weight voting include ballot secrecy, receipt-freeness, and coercion resistance, in order of increasing strength \cite{ali2016overview}. 
Ballot secrecy, a common protection in in-person voting, ensures that individual votes remain hidden from observers. Early cryptographic methods achieved this property~\cite{benaloh1987verifiable,chaum1988elections}. Receipt-freeness prevents voters from proving their choices to others after the fact~\cite{benaloh1994receipt, sako1995receipt, hirt2000efficient}. Coercion resistance achieves a similar property, but in a stronger model of pre-voting interaction with adversaries~\cite{juels2005coercion, grontas2018towards, lueks2020voteagain, achenbach2015improved, dimitriou2020efficient}. These notions and schemes do not immediately generalize to the weighted-voting setting, which thus requires new approaches.

\paragraph{Weighted voting privacy.} 
Eliasson and Z\'{u}chete~\cite{eliasson2006electronic} and Dreier et al.~\cite{dreier2012defining} design secret-ballot schemes for weighted voting, while Cicada~\cite{glaeser2023cicada} does so specifically for blockchain governance, but all of these works disregard tally leakage of ballot contents. Kite~\cite{nazirkhanova2025kite} focuses on private delegation of voting weight. That work acknowledges that tallies may leak private information, briefly mentioning the idea of publishing limited-precision tallies. 

Snapshot's ``shielded voting'' has broad real-world use in blockchain governance. It encrypts ballots, but only ephemerally, decrypting them once proposals conclude~\cite{shieldedvotingsnapshot}. A proposed update offers persistent ballot secrecy~\cite{shutter2023shieldedvoting} via cryptographic techniques and/or trusted execution environments.

\paragraph{Probabilistic tallying methods.} 
Several approaches have explored probabilistic methods for vote tallying to address risks of tally leakage. DPWeVote applies differential privacy to weighted voting in semi-honest cloud environments using randomized response~\cite{Yan07022019}. Random sample voting selects and tallies only subsets of ballots for efficiency~\cite{chaumrsv2016}, with risk-limiting tallies adapting this idea to counter coercion threats~\cite{jamroga2019risk}. Such probabilistic approaches are unacceptable in many voting settings, as small electorates and tight margins may elevate the probability of an incorrect result being reported. Liu et al. analyze privacy under deterministic voting rules using distributional differential privacy, but focus primarily on the unweighted setting and without metrics for attack resistance~\cite{ao2020private}.

\paragraph{Economic modeling of voting behavior.} 
Game-theoretic models are an established approach for analyzing strategic voting behavior. Seidmann shows that when voters may receive external rewards, private voting leads to better organizational outcomes than public voting~\cite{seidmann_theory_2011}, while Bayesian models have been used to examine coordination effects~\cite{fey1997stability}, jury decision-making~\cite{duggan2001bayesian}, shareholder voting~\cite{lee2024shareholder}, and weighted average voting~\cite{renault2007bayesian}. Such games have also been used to examine the robustness of quadratic voting to strategic manipulation, including collusion and fraud~\cite{weyl2017robustness}. Similarly, our B-privacy framework examines the robustness of different tally release mechanisms to bribery in weighted voting systems.

\paragraph{Decentralized autonomous organizations.} 
DAOs represent an important setting for weighted voting and a rich source of real-world voting data. DAOs today provide either no or ephemeral ballot secrecy~\cite{shieldedvotingsnapshot}, but are starting to embrace ballot secrecy in part due to bribery concerns~\cite{aztec2023nounsdao,shutter2023shieldedvoting,vitalikcoinbased}. Empirical studies demonstrate that public visibility creates peer pressure and herding dynamics~\cite{sharma2024unpacking,feichtinger2024sok}, theoretically reducing decentralization~\cite{fabregavoting}. The research community has identified secure voting---specifically, privacy and coercion-resistance---as a critical open problem for DAOs~\cite{tan2023open}, especially in the face of emerging threats like DAOs created for vote-buying, known as Dark DAOs~\cite{mougayar2018darkdao,pereira2024darkdaos}. While these challenges are well-documented, existing work lacks functional metrics to assess how they translate to bribery vulnerabilities. Our B-privacy framework addresses this gap with a quantifiable measure of bribery resistance, and our noise-based mechanism offers a practical approach to enhancing B-privacy while preserving tally transparency.
\section{Voting Framework}\label{sec:prelims}

To study privacy in weighted voting, we first present a general underlying framework used to represent elections and their outcomes. We model a standard voting data format that allows results to be reported in an arbitrary form, so that our definitions generalize to different tallying approaches.

\begin{definition}[Voting transcript]
A \emph{voting transcript} $t = (\mathbf{w},\mathbf{c})$ records the results of a proposal with $n$ voters choosing from a set $C$ of options, where:
\begin{itemize}
\item $\mathbf{w} = (w_1, \ldots, w_n)$ where $w_i \in \mathbb{R}^+$ is the voting weight of voter $i$ and
\item $\mathbf{c} = (c_1, \ldots, c_n)$ where $c_i \in C$ is the option chosen by voter $i$.
\end{itemize}
We write $|\mathbf{s}|$ for the $\ell_1$-norm of sequence $\mathbf{s}$ and for convenience denote the total weight as $W = |\mathbf{w}|$.
\end{definition}

This definition captures scenarios where voters cast all their voting weight for a single option. We do not allow vote splitting in our model.

\begin{definition}[Tally algorithm]
A \emph{tally algorithm} is a (possibly randomized) algorithm $\disclosure$ applied to a voting transcript $t$, where $\disclosure \colon T \to O$ maps transcripts from the space of possible transcripts $T$ to outcomes in outcome space $O$.
\end{definition}

The outcome space $O$ can take various forms depending on the tally algorithm—it might consist of aggregated voting weight for each choice, binary winners, full transcripts, or any other encoding of voting results. To illustrate this definition with a concrete example, consider the raw tally algorithm $\disclosure_\textsf{raw}$, corresponding to the first case, in which transcripts are mapped to aggregated voting weight for each choice.

\begin{example}[Raw tally]\label{example_tally}
Let \( \mathbf{w} = (3.5, 2, 1) \) be voter weights and \( \mathbf{c} = (\textsf{yes}, \textsf{no}, \textsf{yes})\) corresponding choices. Then the raw tally for $t = (\mathbf{w},\mathbf{c})$ is:
\[
\disclosure_\mathsf{raw}(t) = 
\left(
\sum_{i: c_i = \textsf{yes}} w_i,\,
\sum_{i: c_i = \textsf{no}} w_i
\right)
=
(4.5, 2).
\]
That is, the \textsf{yes} option receives total weight 4.5, and \textsf{no} receives total weight 2.
\end{example}

Different tally algorithms offer different privacy-transparency tradeoffs. We summarize the key algorithms used throughout this paper in Table~\ref{tab:result-functions}. Of particular interest is the noised tally algorithm, which we later show can significantly improve privacy while maintaining strong transparency guarantees.

\begin{table}[t]
\centering
\small
\begin{tabular}{@{}ll@{}}
\toprule
\textbf{Algorithm Name} & \textbf{Specification} \\
\midrule
Winner-only &
$\disclosure_\mathsf{winner}(t) = \arg\max_{j \in C} \sum_{i: c_i = j} w_i$ \\[0.8em]

Raw tally &
$\disclosure_\mathsf{raw}(t) = \bigl(\sum_{i: c_i = j} w_i\bigr)_{j \in C}$ \\[0.8em]

Noised tally &
$\disclosure_{\mathsf{noised}(\nu)}(t) = \left(Y_j + \sum_{i: c_i = j} w_i\right)_{j \in C}, \; Y_j \stackrel{\$}{\leftarrow} \nu$ \\[0.8em]

Full-disclosure &
$\disclosure_\mathsf{public}(t) = t$ \\
\bottomrule
\end{tabular}
\caption{Tally algorithms. Each algorithm maps transcript \(t\) to an outcome: 
winner-only returns only the winning choice; raw tally reports the aggregate weight per 
choice; noised tally adds noise \(Y_j\sim\nu\) to each 
choice's tally; full-disclosure returns the complete voting transcript \(t\).}
\label{tab:result-functions}
\end{table}

The raw tally algorithm is the most natural and direct extension of one‑person‑one‑vote privacy to weighted systems. Although~\Cref{prec_attack_ex} demonstrated that privacy can be broken in toy settings, one might hope that larger electorates or real‑world proposals would obscure individual choices. We now show, however, that this is not the case.

\section{Practical Attacks on Raw Tallies}\label{sec:pd_in_daos}

We introduce two attack strategies for extracting individual votes from weighted-voting tallies. We call them a \textit{whale attack} and a \textit{subset sum attack}. We combine these two attacks into a \textit{unified attack algorithm} that efficiently extracts individual ballots given only raw tallies and voter weights.

We explore the efficiency of our unified attack algorithm on data from DAOs. The fact that weighted voting is prevalent in DAOs and individual voting records are publicly available allows us to simulate hypothetical ballot-secrecy scenarios and validate our attacks against ground truth. \footnote{While some systems, such as Snapshot, offer shielded voting as an option, these mechanisms conceal ballots only ephemerally and reveal (anonymous) address voting choices when a proposal vote has concluded.} 

\begin{keyFinding}[title={Key Finding}]
\begin{enumerate}[label=\textbf{(A\arabic*)}]
    \item \textbf{Even under hypothetical ballot secrecy, use of the raw tally algorithm results in significant loss of ballot privacy for DAO voters.}
\end{enumerate}
\end{keyFinding}

\noindent Thus weighted-voting systems cannot directly apply techniques from traditional voting and achieve equivalent privacy.

\subsection{Attack setting and methodology}


We investigate a counterfactual scenario: \textit{If a DAO enforced ballot secrecy but disclosed exact tallies (as is standard in traditional voting), what level of voter privacy would result?} We test our attack algorithm against this scenario by simulating attacks where only aggregate tallies are public, then verifying success against known individual voting records. To do so we use a dataset of votes from the Snapshot voting platform between September 2020 and February 2025 collected for a previous large-scale study of DAO voting\cite{fabregavoting}, limiting our analysis to DAOs with more than 5 proposals, which yields 3,844 proposals across 31 DAOs. 

In our attack model, we assume the set of participating voters and their weights are known (as is typical in token-based systems), but individual vote choices are hidden. Each proposal presents voters with multiple options—typically binary \textsf{yes}/\textsf{no} decisions, although some include additional alternatives. While many proposals offer explicit abstention as a voting option, we focus our analysis on voters who actively cast ballots, treating abstention as a distinct choice only when explicitly available.

We now describe our two complementary attack strategies.

\paragraph{Whale attack.}

The term \textit{whale} denotes large token holders in cryptocurrency systems---and high-weight voters in DAOs. Our \textit{whale attack} exploits the simple but effective observation that if a whale's weight exceeds the total votes for some choice—the whale couldn't have backed that choice.

Consider $\disclosure_\mathsf{raw} = (s_1, s_2, \ldots, s_{|C|})$ where $s_j$ represents the total weight for choice $j \in C$. If $w_i > s_j$ for some voter $i$ and choice $j$, then $c_i \neq j$.

Our whale attack is an \textit{iterative} algorithm: after identifying and removing a whale's vote from one choice, we recompute the remaining tallies and apply the attack again. This process can cause tallies to flip—if enough weight is removed from the initially winning choice, a different choice may become the winner, enabling further whale identification. The algorithm continues until no more whales can be identified, often revealing a substantial fraction of the electorate by weight. It runs in $O(n \cdot |C|)$ time, making it efficient even for large electorates, and serves as an effective preprocessing step for our more computationally intensive subset sum attack.

\paragraph{Subset sum attack.}

Our \textit{subset sum} attack exploits the precision of raw tallies by searching for vote assignments that produce observed outcome. In many cases, there exists a unique transcript $(\mathbf{w}, \mathbf{c})$ that yields the raw tally $\disclosure_\mathsf{raw}(\mathbf{w}, \mathbf{c})$. Reconstructing the vote assignment $\mathbf{c} = (c_1, \ldots, c_n)$ reduces to finding, for each choice $j \in C$, the subset of voters $\{i : c_i = j\}$ whose weights sum to the total $s_j$.

This generalizes the classic subset sum problem to multiple partitions: given voter weights $\mathbf{w}$ and target values $(s_1, s_2, \ldots, s_{|C|})$ from the tally, partition the voters such that each subset's weight sum matches its corresponding target. While multiple partitions could theoretically produce identical sums (and do when voters have identical weights), the high precision of token weights in DAOs (typically 18 decimal places) makes such collisions overwhelmingly unlikely, ensuring that any discovered solution is almost certainly the correct one. The subset sum variant underlying our attack algorithm is NP-hard in general, but two factors make it tractable using well-studied algorithms: (1) the precision mentioned above and (2) the small electorates common in DAO governance.

Lagarias and Odlyzko~\cite{lagarias1985solving} propose an expected polynomial-time algorithm, but it works only for low-density instances, a requirement not satisfied by most DAO voting weight distributions. Instead, we employ the meet-in-the-middle approach of Horowitz and Sahni~\cite{10.1145/321812.321823}, which splits  voter weights into two halves, enumerates all partial sums in each half, and searches for combinations that satisfy the target constraints. This algorithm has $O^*(2^{n/2})$ time and space complexity, making it practical for $n \lesssim 45$ voters on commodity hardware—a significant improvement over the naive $O^*(2^n)$ brute-force approach.

To recover individual votes, we solve a modified problem for each voter $i$: we attempt to find a valid partitioning of $\{w_j : j \neq i\}$ such that adding $w_i$ to the subset for choice $c$ yields the target sum $s_c$. If exactly one choice $c$ allows a valid solution, then voter $i$ must have chosen $c$.

Preprocessing using the whale attack reduces the effective problem size from $n$ to some $n' \leq n$ voters and the complexity of the subset sum instance from $O^*(2^{n/2})$ to $O^*(2^{n'/2})$. This often shrinks problem instances to tractable size.

\smallskip

Our unified attack algorithm is specified as~\Cref{alg:deanon‐balanced}.

\begin{algorithm}[t]
\begin{algorithmic}[1]
  \State \textbf{Input:} Voters $\{1, \ldots, n\}$, weights $\mathbf{w} = (w_1, \ldots, w_n)$, raw tally $\disclosure_{\textsf{raw}}(t) = (s_1, \ldots, s_{|C|})$
  \State $\mathbf{d} \gets \mathbf{0}$ \Comment{Determined votes vector}
  \State $U \gets \{1, \ldots, n\}$ \Comment{Undetermined voters}
  \While{true} \Comment{Whale Attack}
    \State $j^* \gets \arg\max_{j \in C} s_j$ \Comment{Winning choice}
    \State $s_2 \gets$ 2nd‐largest value in $(s_1, \ldots, s_{|C|})$
    \State $W \gets \{i \in U \mid w_i > s_2\}$ \Comment{Whales}
    \If{$W = \varnothing$} \textbf{break} \EndIf
    \ForAll{$i \in W$}
      \State $d_i \gets j^*$ \Comment{Voter $i$ must have voted for $j^*$}
      \State $U \gets U \setminus \{i\}$
      \State $s_{j^*} \gets s_{j^*} - w_i$ \Comment{Update remaining tally}
    \EndFor
  \EndWhile
  \If{$|U| \leq 45$} \Comment{Subset Sum Attack}
    \ForAll{$i \in U$}
      \State Split $U \setminus \{i\}$ into two halves $L, R$
      \State $L_s \gets \{\sum_{k \in S} w_k \mid S \subseteq L\}$
      \State $R_s \gets \{\sum_{k \in S} w_k \mid S \subseteq R\}$
      \State $Q \gets \{j \in C \mid \exists \ell \in L_s, r \in R_s : \ell + r + w_i = s_j\}$
      \If{$|Q| = 1$}
        \State let $j^*$ be the sole element of $Q$
        \State $d_i \gets j^*$ \Comment{Voter $i$ must have voted for $j^*$}
        \State $U \gets U \setminus \{i\}$
        \State $s_{j^*} \gets s_{j^*} - w_i$ \Comment{Update remaining tally}
      \EndIf
    \EndFor
  \EndIf
  \State \Return $\mathbf{d}$
\end{algorithmic}
\caption{Unified Attack Algorithm for Extracting Ballots from a raw tally $\disclosure_\mathsf{raw}(t)$}
\label{alg:deanon‐balanced}
\end{algorithm}

\subsection{Results}\label{result:attack}

We applied our unified attack algorithm (\Cref{alg:deanon‐balanced}) to 3,844 proposals across 31 DAOs. The attack breaks plausible deniability—meaning that it definitively identifies at least one voter's choice—on 3,118 proposals (81.0\%) and recovers \textit{all} voter choices on 1,122 proposals (29.2\%). Among these 3,118 vulnerable proposals, the attack leaks on average 41.6\% of ballots and 85.1\% of voting weight per proposal.

Figure \ref{fig:attack-aggregate} shows the mean effectiveness across DAOs, revealing two key patterns:

\paragraph{Small DAOs experience near-complete privacy failures.} Among proposals with 45 or fewer voters, we achieve complete vote recovery in 791 out of 878 cases (90.1\%). The remaining 87 proposals resist attack only due to voters with identical weights casting different ballots—a scenario that creates fundamental ambiguity. The cluster of smallest DAOs in the top-right of the plot demonstrates that below a certain size threshold, privacy protections collapse entirely. This reflects the subset sum component succeeding on virtually every small proposal, often aided by whale-attack preprocessing that reduces larger proposals to manageable sizes. In 346 proposals that initially exceeded the 45-voter threshold, whale-attack preprocessing successfully reduced the residual voter count to $\leq$45, allowing the subset sum attack to be used.

\paragraph{Larger DAOs' vulnerability depends on whale concentration.} While DAOs with more voters generally appear more resistant (clustering near the y-axis), significant outliers exist. Large DAOs where the subset sum attack is inapplicable may still be highly vulnerable to whale attacks depending on their voting weight distribution. This whale impact is visible in the region near 0\% ballots leaked —few individual votes are revealed, but the most influential voters are completely exposed, with up to 80\% total voting weight leaked.

\Cref{fig:attack-specific} illustrates variations in attack success across proposals within individual DAOs. The effectiveness of both the whale and subset-sum attacks is clearly visible in the Balancer results (far right panel). For all proposals with fewer than 45 voters, the attack achieves complete success, generally leaking 100\% of voting weight. For proposals exceeding 45 voters, whale-attack preprocessing successfully reduces the problem size below the 45-voter threshold in all but two cases, enabling the subset-sum attack to proceed.

Attack success varies considerably across the other three DAOs, revealing several key patterns. Larger winning margins (shown in yellow/green) consistently lead to higher privacy compromise rates, as whales become easier to identify when one choice receives disproportionately low support. This demonstrates that privacy under raw tallies depends on both electorate size and voting patterns.

Notably, DAO size alone does not determine vulnerability. Despite having roughly 10× more voters than Aavegotchi on average, Arbitrum shows much higher attack success rates, with a greater density of proposals where >60\% of voting weight is leaked. The key difference lies in voting weight distribution: while both DAOs have similar average winning margins, Arbitrum exhibits much more concentrated voting weight. In high-margin proposals (>90\%), whale attacks leak 88.0\% of voting weight in Arbitrum versus only 38.1\% in Aavegotchi, highlighting how voting weight concentration can outweigh electorate size in determining vulnerability. These results demonstrate that raw tally vulnerability depends on a combination of electorate size, winning margins, and voting weight concentration.


\begin{figure}[t!]
    \centering
    \includegraphics[width=1\linewidth]{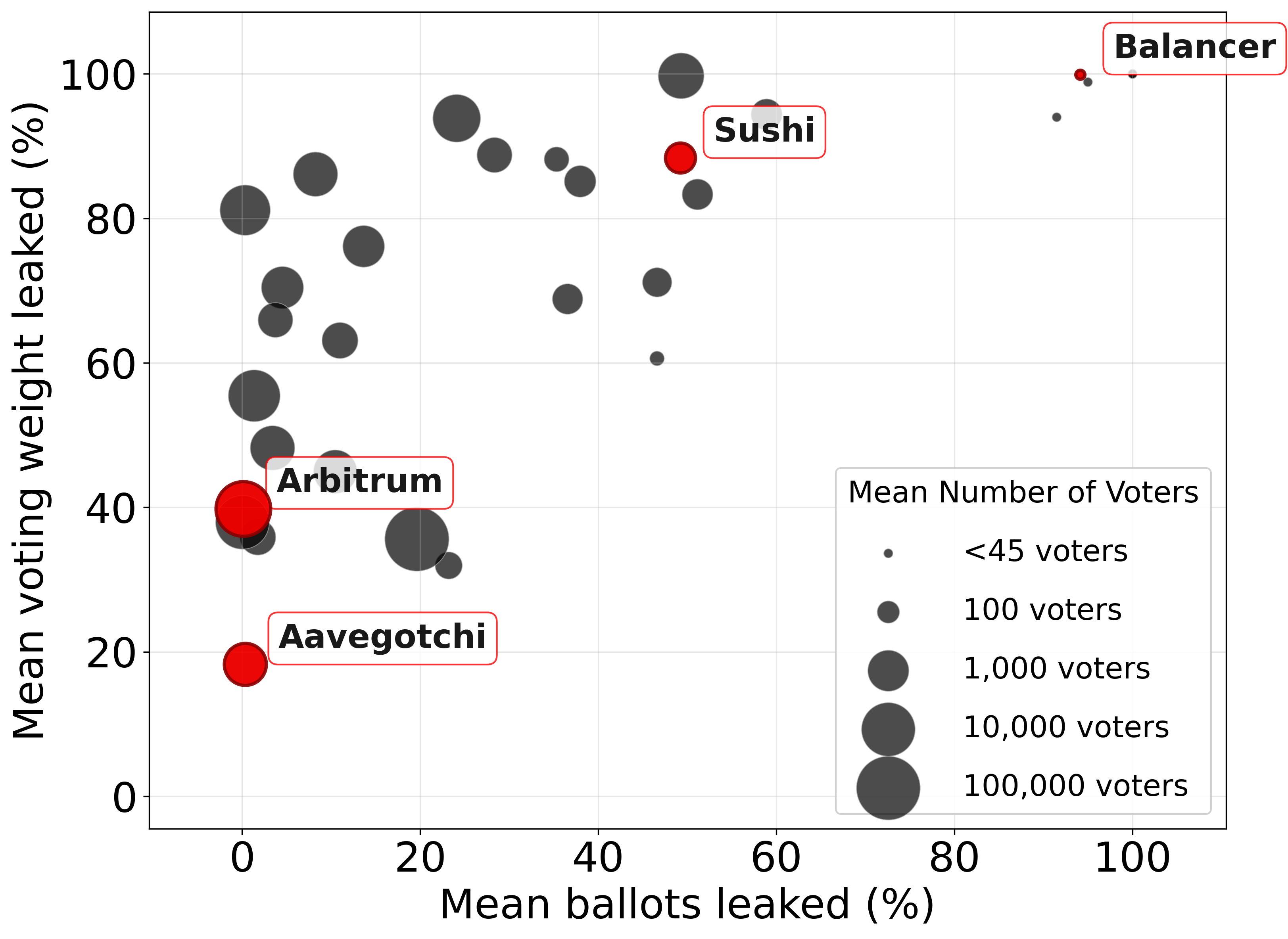}
    \caption{Unified attack algorithm results across all DAOs. Each point represents one DAO and shows mean percentage of ballots leaked (x-axis) vs. mean percentage of voting weight leaked (y-axis) across that DAO's proposals. Point size indicates average voters per proposal. The attacks succeed across diverse DAOs, demonstrating that the raw tally algorithm fails to provide adequate ballot secrecy in weighted voting systems. The attack success for DAOs highlighted in red is presented in more detail in Figure~\ref{fig:attack-specific}.}
    \label{fig:attack-aggregate}
\end{figure}

\begin{figure*}[th!]
    \centering
    \includegraphics[width=\textwidth]{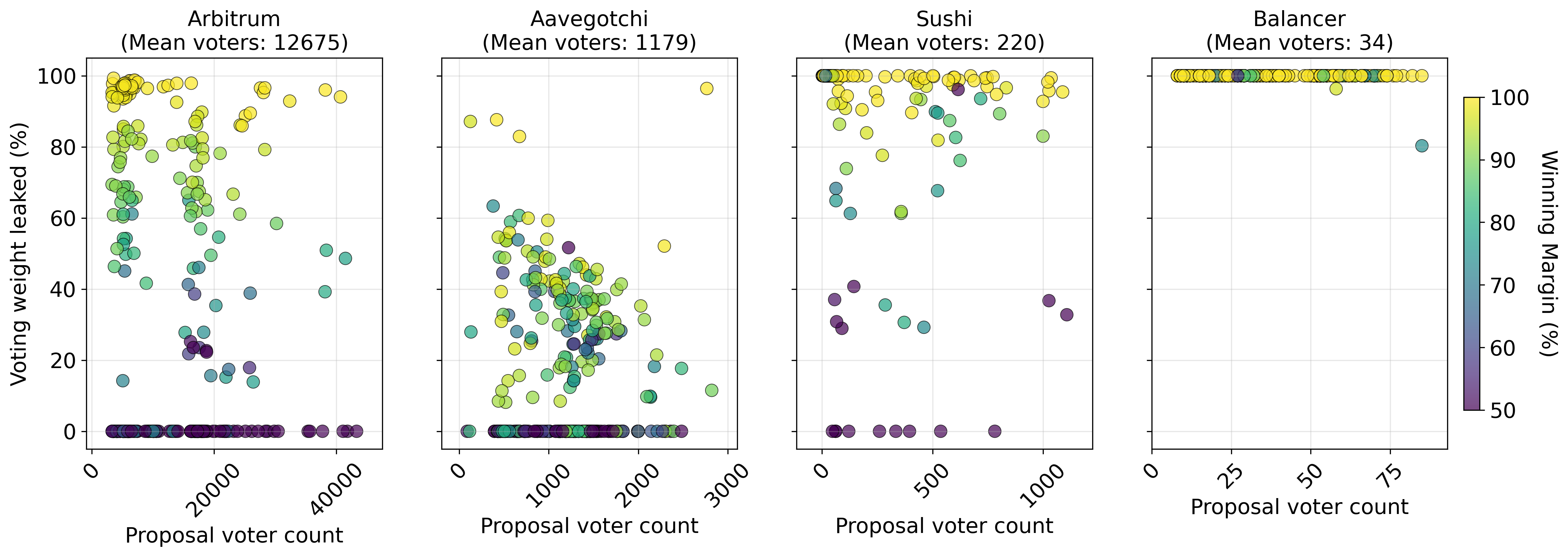}
    \caption{Variations in attack effectiveness across four DAOs with different scales. Each point represents one proposal, colored by winning margin. The view across proposals demonstrates that Attack success is driven by the interaction of electorate size, whale concentration and winning margins—large whales become easier to identify when one choice receives disproportionately low support. Balancer shows near-complete compromise due to small size, while Arbitrum's higher vulnerability compared to Aavegotchi (despite 10× more voters) demonstrates that voting weight concentration can be more important than electorate size for determining privacy risk.}
    \label{fig:attack-specific}
\end{figure*}

\section{B-Privacy: Definition}\label{sec:b-privacy}

In this section, we introduce \emph{B-Privacy}, our new privacy metric for weighted voting. We first present a game-theoretic model 
of adversarial vote‑buying that we call a \textit{bribery game} (\Cref{subsec:bribery_game}). It underpins our subsequent formal definition of B-privacy (\Cref{subsec:B-privacy_definition}). 

\subsection{Bribery game}
\label{subsec:bribery_game}

Our bribery game is a Bayesian game in which voters are rational agents who maximize their expected payoff under uncertainty about other voters' preferences. Each voter must make decisions based on incomplete information about how others will vote, while an adversary strategically offers bribes to influence the outcome.

\paragraph{Voter utilities.}
Our bribery game requires an extension of our voting framework to include voter utilities. For each option $c \in C$, let $U^{c} = (U_1^{c}, \ldots, U_n^{c})$ where $U_i^c \in \mathbb{R}$ represents voter $i$'s utility from outcome $c$. We focus on binary proposals ($C = \{\mathsf{yes}, \mathsf{no}\}$) and normalize by setting $U_i^\mathsf{yes} = 0$, so $U_i^\mathsf{no}$ represents voter $i$'s relative preference for the ``no'' option. We assume all voters participate—abstention could be modeled as an explicit third choice, although we restrict to the binary case for simplicity as it captures the essential tension between tally transparency and B-Privacy while avoiding the additional complexity of preference aggregation across multiple alternatives. We leave extension to multi-choice settings for future work.

Our model for this game must specify what information voters have about other voters' utilities and likely choices. 
We follow the standard approach in Bayesian voting games (e.g. \cite{renault2007bayesian}) and model utilities as private types drawn from commonly known distributions. Each voter $i$'s utility $U_i^\mathsf{no}$ is drawn independently from distribution $\mathbf{U}_i^\mathsf{no}$ with infinite support, where these distributions are public but the realized values are private information.

\paragraph{Coercion-resistance.} Our bribery game assumes the underlying voting system is coercion-resistant: voters cannot cooperate with an adversary to prove how a ballot was cast. Under this assumption, the published tally is the sole information available for designing bribery strategies. B-privacy therefore complements coercion-resistance: the latter prevents verification through direct interaction, while the former limits what can be inferred from the tally alone. Although B-privacy is defined in this idealized setting, it remains practically relevant more generally as it bounds the privacy loss specifically due to tally publication.

\bigskip
\noindent We present the full bribery game in~\Cref{fig:bribery_game}.

\paragraph{Adversarial strategy.} In our bribery game, the adversary commits to bribe amounts $\mathbf{b}$ prior to the proposal. They do not depend on the outcome. The payment conditions $\mathbf{f}$, in contrast, can depend on the outcome. 
This separation cleanly distinguishes the information-theoretic aspects of tally algorithm choice (captured by optimal $\mathbf{f}$) from the economic optimization of bribe amounts ($\mathbf{b}$), enabling the comparison among tally algorithms that is our main focus.



\subsection{B-privacy definition}
\label{subsec:B-privacy_definition}

Using the bribery game we can now define B-Privacy.

\begin{definition}[B-Privacy] \label{def:b-privacy}
For given inputs to the bribery game $(\mathbf{w}, \mathbf{U}^\mathsf{no})$ The \emph{B-privacy} $B_{\disclosure}(p)$ is the minimum total bribe budget required for an adversary to achieve success probability $p$ when tally algorithm $\disclosure$ is used:
$$B_{\disclosure}(p) = \min_{\mathbf{b}, \mathbf{f}} |\mathbf{b}| ~\text{s.t.} \quad p_{{\sf succ}} \geq p.$$ 

\noindent The \emph{relative} B-privacy of tally algorithm $\disclosure$ is the ratio $B_{\disclosure}(p)/B_{\disclosure_{\mathsf{public}}}(p)$.
\end{definition}

B-Privacy is measurable for any tally algorithm. Even $\disclosure_{\mathsf{winner}}$, which only discloses the winner, enables strategic bribery through outcome-conditional payments (e.g., ``payment if $\mathsf{yes}$ wins''). Looking forward, one of our contributions is to identify tally algorithms that maximize transparency while maintaining high B-Privacy. 

\section{B-Privacy: Computation}
\label{sec:B-privacy_comp}

Computing B-privacy, as specified in Definition~\ref{def:b-privacy}, requires computation of optimal adversarial bribery strategies. In this section, we introduce foundational analytic concepts for this purpose (\Cref{subsec:analytic_concepts}) and then present results and methods for computing B-privacy (\Cref{sec:adv_strategy}).

\subsection{Pivotality and bribe margin}
\label{subsec:analytic_concepts}

In analyzing and computing B-privacy, we make use of two foundational concepts: \textit{pivotality} and \textit{bribe margins}. For a given voter $i$, tally algorithm, and adversarial strategy, these concepts respectively reflect the voter's likelihood of affecting the vote outcome given the voter's individual voting choice and of receiving a bribe.

Both quantities are computed from voter $i$'s perspective and assume a Bayesian Nash equilibrium has been reached. Unless otherwise noted, probabilities are taken over the randomness in other voters' utilities $U_j^\mathsf{no} \sim \mathbf{U}_j^\mathsf{no}$ for $j \neq i$, which determine the equilibrium choices $c_j$ of other voters and the resulting voting transcript $t = (\mathbf{w}, \mathbf{c})$.

\begin{definition}[Pivotality]\label{def:pivotality}
For voter $i$, the pivotality $\Delta_i$ measures how much their vote affects the probability of a $\mathsf{no}$ outcome:
$$\Delta_i = \Pr[\mathsf{no} \text{ wins} \mid c_i = \mathsf{no}] - \Pr[\mathsf{no} \text{ wins} \mid c_i = \mathsf{yes}],$$
\end{definition}

Pivotality is a standard concept in voting games~\cite{fey1997stability, duggan2001bayesian, weyl2017robustness}. In our framework, it captures the intuition that voters with lower pivotality are more susceptible to bribes—since their vote is less likely to affect the outcome, the potential bribe becomes relatively more attractive. The bribe margin, by contrast, is specific to our bribery framework and captures how a voter's choice affects their expected bribe payment under the adversarial strategy.

\begin{definition}[Bribe margin]\label{def:bribe-margin}
For voter $i$ with bribery condition function $f_i$, the bribe margin $\bm_i$ is the additional probability of receiving a bribe when voting $\mathsf{yes}$ versus $\mathsf{no}$:
$$\bm_i = \Pr[f_i(\disclosure(t)) = 1 \mid c_i = \mathsf{yes}] - \Pr[f_i(\disclosure(t)) = 1 \mid c_i = \mathsf{no}].$$
When $\disclosure$ is non deterministic, the probability is also over this additional source of randomness.
\end{definition}

\begin{figure}[t!]
    \centering
\begin{tcolorbox}[title = Bribery Game, label=box:bribery-game]
\footnotesize

\textbf{Game Inputs}: 
\begin{itemize}
\item Voter weights $\mathbf{w} = (w_1, \ldots, w_n)$ 
\item Utility distributions $\mathbf{U}^\mathsf{no} = (\mathbf{U}_1^\mathsf{no}, \ldots, \mathbf{U}_n^\mathsf{no})$
\item Tally algorithm $\disclosure$
\end{itemize}

\textbf{Players}: $n$ voters (indexed by $i \in \{1, \ldots, n\}$) and one adversary

\smallskip

\textbf{Adversary's Objective}: The adversary seeks to maximize the probability of a $\mathsf{yes}$ outcome.

\smallskip

\textbf{Information Structure:} 
\begin{itemize}
\item Private utilities $U_i^\mathsf{no} \stackrel{\$}{\leftarrow} \mathbf{U}_i^\mathsf{no}$ are drawn for each voter $i$
\item Each voter $i$ observes only their own private utility $U_i^\mathsf{no}$
\item The distributions $\mathbf{U}_1^\mathsf{no}, \ldots, \mathbf{U}_n^\mathsf{no}$ are common knowledge among all players
\end{itemize}

\textbf{Game Sequence}:
\begin{enumerate}
\item \textit{Adversary's strategy}: Adversary commits to bribe amounts $\mathbf{b} = (b_1, \ldots, b_n)$ and bribery condition functions $\mathbf{f} = (f_1, \ldots, f_n), ~f_i: O \to \{0,1\}$ where $O$ is the set of possible outcomes, and $f_i(\disclosure(t)) = 1$ indicates voter $i$ receives bribe $b_i$ given voting transcript $t=(\mathbf{w},\mathbf{c})$

\item \textit{Voting stage}: Voters simultaneously choose votes $\mathbf{c}$ to maximize expected utility
\item \textit{Outcome}: Tally algorithm $\disclosure$ reveals outcome information and bribes are paid per condition functions
\end{enumerate}

\textbf{Equilibrium}: Bayesian Nash equilibrium where each voter's strategy maximizes their expected utility given their beliefs about others' behavior, and these beliefs are consistent with equilibrium play.

\end{tcolorbox}
    \caption{Bribery game underpinning B-privacy definition.}
    \label{fig:bribery_game}
\end{figure}

These quantities determine voter $i$'s decision as shown in the following theorem:

\begin{theorem}[Adversary's success probability]\label{thm:adv_success_prob}
Given bribe vector $\mathbf{b}$ and bribery condition functions $\mathbf{f}$, the adversary's probability of achieving a $\mathsf{yes}$ outcome is 

\[
p_{{\sf succ}} = \Pr\left[\sum_{i=1}^n w_i X_i > W/2\right],
\]

\noindent where $X_i = \mathbb{I}[U_i^{\sf no} \leq \frac{\bm_i b_i}{\Delta_i}]$ for $U_i^{\sf no} \stackrel{\$}{\leftarrow} \mathbf{U}_i^{\sf no}$ is an indicator random variable and voting behavior follows the Bayesian Nash equilibrium induced by $(\mathbf{b}, \mathbf{f})$.

\end{theorem}

\begin{proof}
    Under bribe vector $\mathbf{b}$ and bribery condition functions $\mathbf{f}$, 
    voter $i$'s expected utility $\mathbb{E}[U_i]$ given they vote $\mathsf{yes}$ is:
    $$\Pr[\mathsf{no} \text{ wins} | c_i = \mathsf{yes}] \cdot U_i^\mathsf{no} + b_i \cdot \Pr[f_i(\disclosure(t)) = 1 | c_i = \mathsf{yes}].$$

    And given they vote $\mathsf{no}$ is:
    $$\Pr[\mathsf{no} \text{ wins} | c_i = \mathsf{no}] \cdot U_i^\mathsf{no} + b_i \cdot \Pr[f_i(\disclosure(t)) = 1 | c_i = \mathsf{no}].$$

    Accordingly, in equilibrium, voter $i$ votes yes if and only if 
    \begin{align*}
        \mathbb{E}[U_i | c_i = \mathsf{no}] - \mathbb{E}[U_i | c_i = \mathsf{yes}] &\leq 0\\
        \Delta_i U_i^\mathsf{no}-\bm_ib_i &\leq 0 \\
        U_i^\mathsf{no} &\leq \frac{\bm_i b_i}{\Delta_i}.
    \end{align*}
    The second line follows by substituting the definitions of pivotality $\Delta_i$ and bribe margin $\bm_i$.

    Every voter's private utility is drawn as $U_i^{\sf no} \stackrel{\$}{\leftarrow} \mathbf{U}_i^{\sf no}$, so a voter's contribution to the $\mathsf{yes}$ total is the random variable $w_iX_i$ where $X_i = \mathbb{I}[U_i^{\sf no} \leq \frac{\bm_i b_i}{\Delta_i}]$. The adversary succeeds when the $\mathsf{yes}$ total exceeds $W/2$, which occurs with probability
    
    \[
    p_{{\sf succ}} = \Pr\left[\sum_{i=1}^n w_i X_i > W/2\right].
    \]

\end{proof}

The key to analyzing B-Privacy for a given bribery strategy $(\textbf{b},\textbf{f})$ lies in computing the equilibrium values of bribe margin $\bm_i$ and pivotality $\Delta_i$. We show later for the tally algorithms we consider, bribe margins (or bounds on them) can be expressed in terms of pivotality, which simplifies analysis. 


\subsection{Optimal adversarial strategy}\label{sec:adv_strategy}

To compute B-Privacy, we must solve the optimization / minimization problem in Definition \ref{def:b-privacy}.

A key insight is that we can restrict attention to strategies where condition functions operate independently across voters—that is, voter $i$'s payment depends only on the outcome, not on correlations with other voters' payments. While correlated payment schemes might seem more powerful, we show in Appendix \ref{app:optimal_strategy_indep} that any correlated strategy can be replaced by an equivalent independent strategy with the same B-privacy. The intuition is that each voter cares only about their own probability of receiving a bribe under different outcomes. Correlating payments across voters doesn't change these individual probabilities, so it provides no advantage to the adversary. This allows us to focus on the simpler case where $\mathbf{f} = (f_1, \ldots, f_n)$ for independent condition functions.


Given this simplification, our goal becomes characterizing the adversary's optimal choices of bribe amounts $\mathbf{b}$ and condition functions $\mathbf{f}$ that, per the definition of B-privacy, minimize total cost while achieving success probability $p$. We refer to the values of $\mathbf{b}$ and $\mathbf{f}$ that solve this B-privacy optimization problem $\min_{\mathbf{b}, \mathbf{f}} |\mathbf{b}|$ as \emph{optimal}.

We show in~\Cref{thm:optimal-bcf} that the optimal strategy has a simple structure: the adversary should choose condition functions that maximally exploit the information revealed by the tally algorithm, then set bribe amounts to achieve the target success probability at minimum cost.


As with Definitions \ref{def:pivotality} and \ref{def:bribe-margin}, probabilities are taken over other voters' utilities $U_j^\mathsf{no} \sim \mathbf{U}_j^\mathsf{no}$ for $j \neq i$, which determine their equilibrium choices.

\begin{theorem}[Optimal bribery condition functions]\label{thm:optimal-bcf}
The optimal condition function for voter $i$ is the one that maximizes bribe margin $\bm_i$. Define $p_{i,c}^{o} = \Pr[\disclosure(t) = o \mid c_i = c]$. Then the optimal condition function takes the form:
$$f^*_i(o) = \mathbb{I}\{p_{i,\mathsf{yes}}^{o} \geq p_{i,\mathsf{no}}^{o}\}$$
and corresponds to optimal bribe margin:
$$\bm_i^* = \sum_{o \in O} \max(p_{i,\mathsf{yes}}^{o} - p_{i,\mathsf{no}}^{o}, 0).$$
\end{theorem}

The proof is deferred to Appendix~\ref{app:optimal_strategy_bcf}. This theorem tells us that adversaries should condition bribes on outcomes that serve as strong evidence of voter compliance. Specifically, the adversary pays voter $i$ when the observed outcome $o$ is more likely under the scenario where voter $i$ voted $\mathsf{yes}$ than under the scenario where they voted $\mathsf{no}$. The resulting bribe margin $\bm_i^*$ quantifies how much this optimal conditioning improves the adversary's ability to target bribes. Theorem~\ref{thm:adv_success_prob} shows that higher bribe margins allow the adversary to achieve the same success probability with lower total bribe costs.

\paragraph{Computing $B_{\disclosure}(p)$.} Given optimal condition functions $\mathbf{f}^*$, computing $B_{\disclosure}(p)$ requires solving two interconnected problems: (1) determining the Bayesian Nash equilibrium and (2) finding the optimal bribe allocation.

\textbf{Problem 1: Determining the Bayesian Nash equilibrium.} We do not attempt to characterize the equilibrium analytically, and instead use a computational approach to find the fixed point of voter pivotalities $\boldsymbol{\Delta}$. We iterate until convergence: given current pivotalities, we compute voter choice probabilities, which in turn determine new pivotalities.

\textbf{Problem 2: Optimal bribe allocation.} This optimization problem is challenging because the optimal bribe amounts $\mathbf{b}$ depend on the equilibrium pivotalities, but changing the bribes alters the equilibrium itself, creating a circular dependency. Moreover, bribe margins $\boldsymbol{\alpha}$ for some tally algorithms also depend on the pivotalities, adding further complexity.

Our solution is to decouple these problems by fixing reasonable bribe allocation strategies upfront, then computing the resulting equilibrium for each strategy. This approach is well-motivated for several reasons. First, we expect that reasonable allocation strategies differ minimally in terms of optimality and resulting B-privacy values—the loss in theoretical optimality is likely small compared to the uncertainties introduced by our modeling assumptions. Second, our goal is to identify broad relationships and trends rather than exact numerical predictions, so small differences in allocation optimality are less consequential than understanding how tally algorithms affect bribery costs. Finally, the simplifying assumptions in our model make the precise optimization problem less meaningful than exploring heuristic strategies that better reflect realistic adversarial behavior.

We test several bribe allocation strategies an adversary might employ, focusing bribes on voters who are unlikely to support the adversary's preferred outcome (since bribing supporters would be wasteful). These include: even distribution across opposing voters, weighting by voter weight among opponents (quadratically, logarithmically, etc.), and targeted approaches focusing on the largest opposing voters. For each allocation strategy, we adopt an iterative method with full details provided in Appendix~\ref{app:comp-methods}:
\begin{enumerate}[noitemsep]
    \item Allocate bribes according to the chosen strategy subject to budget constraint $\sum_i b_i = B$.
    \item Compute the resulting equilibrium by iterating until pivotality convergence:
    \begin{enumerate}[noitemsep]
        \item Compute bribe margins and voter choice probabilities using current pivotalities.
        \item Update pivotalities based on resulting bribe margins and voter choice probabilities.
    \end{enumerate}
\end{enumerate}

For each strategy, we perform a binary search over budgets to identify the minimal budget~$B$ that yields adversarial success probability~$p$. Among these, we select the strategy with the lowest budget, providing a practical approximation of B-privacy, $B_{\disclosure}(p)$.

\paragraph{Summary: computing B-privacy.} To compute $B_{\disclosure}(p)$ for a given tally algorithm $\disclosure$ and probability $p$, we must solve the optimization problem specified in Definition~\ref{def:b-privacy}. While optimal condition functions $\mathbf{f}^*$ have the analytical characterization of Theorem~\ref{thm:optimal-bcf}, we use heuristic bribe allocation strategies that might be adopted by a realistic adversary to decouple bribe allocation from the equilibrium, and then approximate this equilibrium using an iterative approach. We apply this approach to real DAO data in Section~\ref{sec:noise_experiments}, and find empirically that it converges reliably across all tally algorithms and proposals in our analysis.

\subsection{B-Privacy and Plausible Deniability}

To connect B-Privacy with classical privacy notions, we examine its relationship to \emph{plausible deniability}, a privacy concept that captures whether individual choices can be inferred from disclosed information. We adapt this notion to our weighted voting and bribery game setting.

Within our bribery game framework, plausible deniability captures the adversary's uncertainty about a voter's choice: given the tally outcome, it measures how much the adversary's confidence in their inference exceeds their prior belief about that voter's behavior. The normalization by the prior probability is essential—it distinguishes between information genuinely revealed by the tally versus voters who were already predictable based on the adversary's prior knowledge in our model. We take the minimum across both choices to capture the worst-case scenario: plausible deniability is compromised if the adversary can confidently rule out either choice, regardless of which choice the voter actually made.

\begin{definition}[Plausible Deniability]\label{def:plausible}
    For voter $i$ and tally algorithm $\disclosure$, let $p^{i,c}_{o} = \Pr[c_i = c|\disclosure(t)=o]$. We define the plausible deniability as:
    \[
        PD_i^{\disclosure}(o)=\min\left(\frac{p^{i,\yes}_{o}}{\Pr[c_i=\yes]}, \frac{p^{i,\no}_{o}}{\Pr[c_i=\no]}\right).
    \]
    The expected plausible deniability of voter $i$ under tally algorithm $\disclosure$ is:
    \[
        EPD_i^{\disclosure}=\mathbb{E}_{o \sim \disclosure(t)}[PD_i^{\disclosure}(o)].
    \]
    where the expectation is over the distribution of possible outcomes induced by the randomness in voter utilities in our bribery game model.
\end{definition}

This definition directly captures the privacy failures demonstrated in Section~\ref{sec:pd_in_daos}. When our attacks definitively identified a voter's choice—for instance, determining that a whale must have voted yes because their weight exceeds the total no votes—the voter has zero plausible deniability under this definition. Specifically, if the tally outcome $o$ allows the adversary to conclude with certainty that voter $i$ chose option $c$, then $p^{i,c}_{o} = 1$ and $p^{i,c'}_{o} = 0$ for $c' \neq c$, yielding $PD_i^{\disclosure}(o) = 0$.


Plausible deniability directly relates to the bribe margin $\bm_i$, which in turn determines B-Privacy:

\begin{theorem}\label{lem:bribe-margin-epd}
    For voter $i$ in the Bribery Game with tally algorithm $\disclosure$,
    \[
    EPD_i^{\disclosure} = 1-\bm_i^*,
    \]
    where $\bm_i^*$ is the optimal bribe margin for voter $i$.
\end{theorem}

The proof can be found in Appendix~\ref{app:bound_PD}. This equation demonstrates the inverse relationship between bribe margin and plausible deniability: as the adversary's ability to condition bribes on outcomes increases (higher $\bm_i^*$), the voter's privacy decreases (lower $EPD_i$).

\section{Noised Tally Algorithms}
\label{sec:noise_mechanism}

The attacks presented in~\Cref{sec:pd_in_daos} show that releasing raw tallies undermines user privacy---in many cases stripping users even of plausible deniability, i.e., leaking their exact voting choices. This ability by adversaries drives down B-privacy.

In this section we explore the mechanism mentioned  in~\Cref{sec:prelims}:  \textit{noised} tally algorithms. We analyze how their properties can improve B-privacy while maintaining a high level of transparency in resulting tallies. 

As with our previous analyses, for simplicity we assume binary choice proposals. The goal then is to tally the weight that voted $\mathsf{yes}$ (the remaining weight must have voted $\mathsf{no}$) with added noise from some distribution $\nu$:
\[
\disclosure_{\mathsf{noised}(\nu)}(t) = Y+\sum_{i: c_i = \mathsf{yes}} w_i ,\quad Y \stackrel{\$}{\leftarrow} \nu.
\]

A subtle but key technical issue here is that noising may yield a tally that \textit{indicates the incorrect winner}. As long as the support of the noise distribution includes values greater than the margin, there is a risk of incorrectly flipping the result. Tightly bounding $\nu$ might seem to address this problem, but may unduly erode the benefit of noising and will still induce an incorrect result with non-zero probability. We instead consider a simple solution: Specify the winner in the tally in addition to the noised result. We call this a \textit{corrected} noised tally algorithm. Denoted by $\disclosure_{\mathsf{noised}(\nu)+}$, it is as follows:
\[\disclosure_{\mathsf{noised}(\nu)+}(t)=\left (\disclosure_{\mathsf{noised}(\nu)}(t), \disclosure_{\mathsf{winner}}(t)\right).\]

Our noising approach  appears similar to that used for differential privacy. Indeed, when using an \textit{uncorrected} noised tally with Laplacian noise, we show in~\Cref{app:dp-disclosure} that we can use differential privacy to upper-bound the advantage of any bribery condition function. However, \textit{differential privacy guarantees break down when the correct winner is revealed}. This motivates our derivation of an upper bound on bribe margins for corrected noised tallies.

\begin{theorem}[Bound on corrected noised tally bribe margin]\label{thm:noised-bribe-margin}
    When the corrected noised tally $\disclosure_{\mathsf{noised}(\nu)+}$ is used,  the optimal bribery condition function for any voter $i$ has bribe margin at most 
    \[
    \Delta_i + (1-\Delta_i)\int_{-\infty}^{\infty} \max\left(\Pr[Z=u-w_i] - \Pr[Z=u], 0 \right)\,du, 
    \]
    
    \noindent for random variable $Z \sim \nu$.
\end{theorem}

The proof is deferred to Appendix~\ref{app:optimal_strategy_noisy}.

\paragraph{Intuition.}  The bribe margin decomposes based on whether voter $i$ is pivotal. When pivotal (with probability $\Delta_i$), voter $i$'s vote determines the winner, giving the adversary perfect conditioning ability. When not pivotal (with probability $1-\Delta_i$), both choices yield the same winner, so the adversary must use the noised tally to distinguish between them. The integral term measures how much the noise distributions overlap when voter $i$'s weight shifts the tally by $w_i$. Greater overlap means less adversarial advantage from the noised information.

\Cref{thm:noised-bribe-margin} result reveals a fundamental tradeoff in noised tally algorithms. While outputting the winner guarantees correctness of the final decision, transparency depends critically on the noise distribution chosen. Adding more noise improves B-Privacy (by reducing the integral term) but degrades the fidelity of the raw tally, limiting its usefulness for understanding margins and community consensus.

In Section~\ref{sec:noise_experiments}, we explore navigation of this tradeoff in practice—that is, whether we can achieve strong relative B-Privacy with sufficiently low error in the raw tally to preserve meaningful transparency.

\subsection{Calibrating noise}
\label{sec:calibrating-noise}

In our Section~\ref{sec:noise_experiments} experiments we use a Laplacian noise distribution $\nu = \mathrm{Lap}(0,b)$ for $\disclosure_{\mathsf{noised}(\nu)+}$. For a given proposal, we calibrate the parameter $b$ in a way that helps with understanding how tallies are perturbed and thus the impact on transparency, as follows. 

Let the total voting weight for the proposal be $W$. For a noise distribution $\nu$, we define the \textit{tally perturbation} with frequency $q$ as a percentage $\talprtb\%$ such that:
\begin{equation*}
    \Pr(|Y|\le \talprtb W) = q, \mbox{ for }Y \stackrel{\$}{\leftarrow} \mathcal{\nu}.
\end{equation*}

In other words, a $q$-fraction of the time, the noise $Y$ drawn from $\nu$ has magnitude $\leq dW$. In our experiments, we set the frequency $q = 0.95$. We explore $\talprtb \in [0\%,100\%]$, but in practical applications would typically choose small values, e.g., $\talprtb = 10\%$. For such parameters, 95\% of the time, $|Y| \leq 0.1 W$, so that the published tally differs from the true, raw tally, by at most 10\%. 

For a chosen frequency $q$ and value $d$ of tally perturbation, we set $b$ for noise $\nu = \mathrm{Lap}(0,b)$ to achieve it. For Laplace noise $Y\sim \mathrm{Lap}(0,b)$, we need $\Pr(|Y|\le dW) = q$. Since $\Pr(|Y|\le x)=1-e^{-x/b}$, this gives us $1-e^{-dW/b} = q$, which solves to $b = \frac{dW}{-\ln(1-q)}$. For our experimental choice of $q = 0.95$, this simplifies to $b=\frac{dW}{\ln 20}$.

An example helps illustrate.

\begin{example}[Tally perturbation]
\label{ex:TP}
Consider a proposal with:
\begin{itemize}
    \item Weights: {$\mathbf{w} = (1,\, 1.1,\, 1.3)$} (\mbox{thus }$W = 3.4$)
    \item Voter choices: $\mathbf{c} = (\textsf{yes}, \textsf{no}, \textsf{yes})$
    \item Raw tally: $T = (\mbox{\textsf{no: }} 1.1, \mbox{\textsf{yes: }} 2.3)$ 
\end{itemize} 

\noindent Suppose $q = 0.95$ and \talprtb = 10\%. Then computation of the corrected noised tally might look like this:

\smallskip

\begin{itemize}
\item Noise distribution: $\nu = \mathrm{Lap}(0,b)$ for $b =0.1 W/\ln 20$
\item Noise draw: $Y = 0.1 \sim \mathcal{\nu}$
\item Noised tally: $\tilde{T} = (\mbox{\textsf{no: }} 1.2 , \mbox{\textsf{yes: }} 2.2 )$
\item Corrected noised tally: $(\mbox{Winner:} \textsf{yes}, \tilde{T})$
\end{itemize}

\noindent Note that the noise magnitude, $0.1 < \talprtb W = 0.1 \times 3.4 = 0.34$, as expected $95\%$ of the time ($q = 0.95$). Note also that while individual voter choices can be deduced with certainty from the raw tally, they cannot be from the noised tally. (E.g., which voter voted \textsf{no}?)
\end{example}

\subsection{Practical attacks on noised tallies}
\label{sec:noised_tally_attacks}

To provide practical intuition for how noise limits adversarial capabilities, we adapt our raw tally attacks from~\Cref{sec:pd_in_daos} to the noised setting. Consider an adversary who receives a corrected noised tally calibrated to 10\% tally perturbation with 95\% frequency as defined in~\Cref{sec:calibrating-noise}.

Noise fundamentally breaks the subset sum attack, since it's extremely unlikely that any voter partitioning produces the exact noised tally values. The whale attack, however, remains partially viable with modifications.

Under raw tallies, if a whale's weight $w_i$ exceeds the total votes $s_j$ for some choice $j$, then certainly $c_i \neq j$. In a noised tally with 10\% tally perturbation, this logic requires adjustment: to maintain at least 95\% confidence, the adversary can only conclude $c_i \neq j$ if $w_i > s_j + 0.1W$, accounting for the possibility that noise increased the observed tally for choice $j$ by up to $0.1W$ from its true value. Additionally, when the true winner is released (as it is in our corrected noised tally algorithm) the adversary can always conclude $c_i = j$ if $w_i > \frac{W}{|C|}$ (the voter is a majority whale) and $\text{tally}_{\text{winner}}(t) = j$.

This adapted whale attack—which we implement by modifying the threshold condition in our algorithm—represents a conservative adversarial strategy that maintains high confidence in leaked votes while acknowledging the uncertainty introduced by noise. To evaluate how noise affects attack effectiveness, we apply Laplace noise calibrated to tally perturbation of 10\% to every proposal and run the modified whale attack algorithm on the resulting noised tallies, repeating this process across 10 trials per proposal to account for the randomness in noise generation. Figure~\ref{fig:attack-noised} shows the dramatic reduction in attack effectiveness averaged across the 10 trials: the same DAOs that suffered near-complete privacy compromise under raw tallies achieve substantially better privacy protection with noised tallies.

The results illustrate why adding noise creates practical challenges for adversaries. Although a more aggressive adversary might tolerate greater uncertainty by adopting smaller thresholds, such strategies come at the cost of the confidence that makes targeted bribery effective. The limitations of the attack considered here are consistent with our theoretical bound from Theorem~\ref{thm:noised-bribe-margin}, showing how noise yields concrete privacy benefits by forcing adversaries to operate under uncertainty.

\begin{figure}[t!]
    \centering
    \includegraphics[width=1\linewidth]{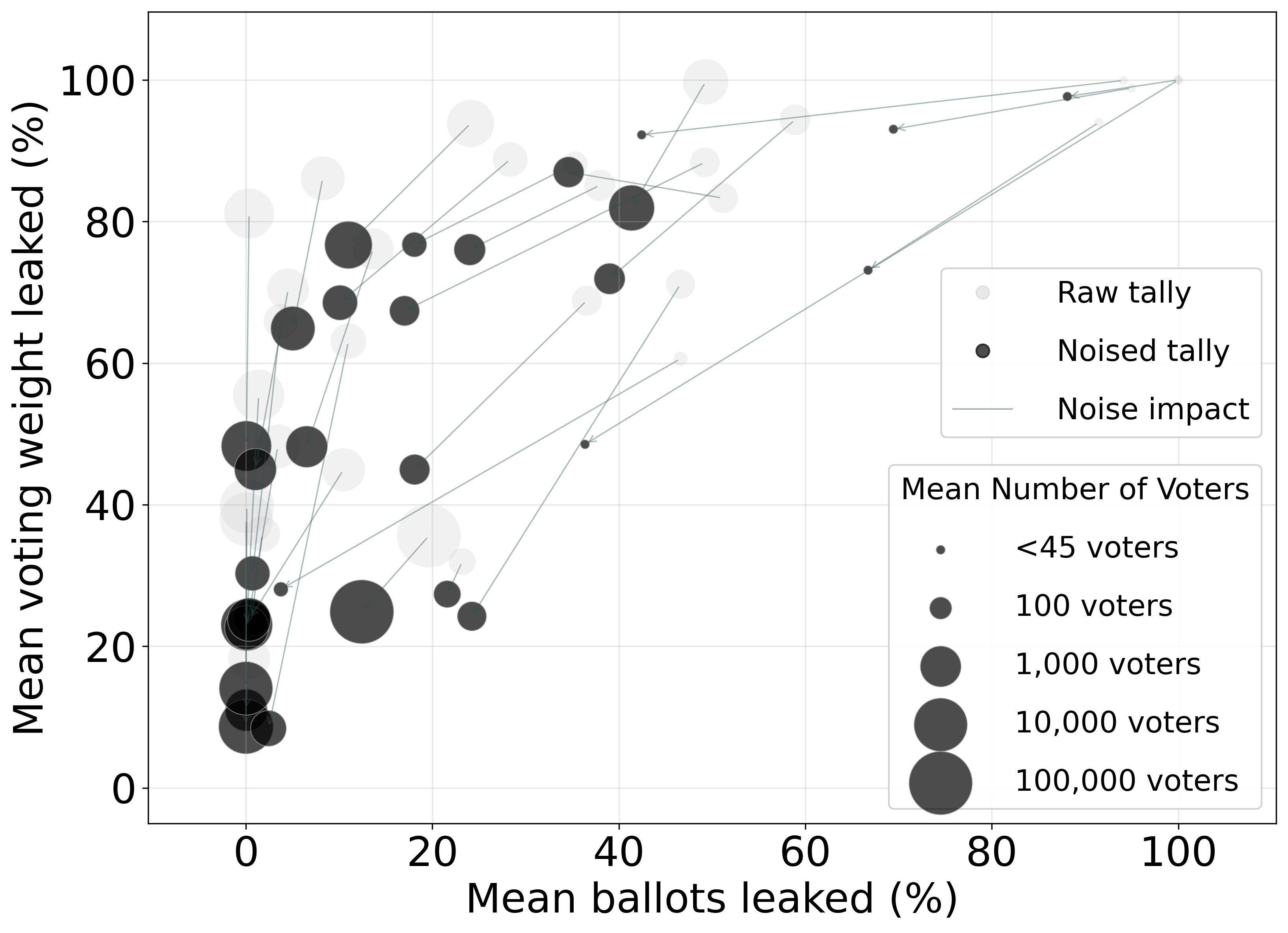}
    \caption{Comparison of raw tally versus adapted noised tally attack effectiveness across all DAOs when 10\% tally perturbation is applied. Setup mirrors Figure~\ref{fig:attack-aggregate}, but shows how each DAO's position shifts when attacks are adapted for noised tallies. The leftward and downward shift demonstrates greatly reduced attack effectiveness under noise.}
    \label{fig:attack-noised}
\end{figure}

\section{Empirical Analysis of B-Privacy}
\label{sec:noise_experiments}

Our theoretical framework characterizes how adversarial knowledge, weight distributions, and tally algorithms influence B-Privacy. To validate these insights and explore the practical privacy-transparency tradeoffs, we conduct a simulation study that uses historical DAO voting data to model realistic weighted voting scenarios under our B-privacy framework. Specifically, we seek to understand in practice the extent to which noised tallying can raise the cost of bribery without undermining transparency, and how intrinsic election properties such as decentralization mediate this trade-off.

Unlike our earlier empirical evaluation of privacy attacks (Section~\ref{sec:pd_in_daos}), which directly analyzed historical voting records, this analysis requires simulation because B-privacy depends on unobservable voter utilities. We use the historical voting patterns to infer plausible voter utility distributions, then simulate how these voters would behave under different tally algorithms when facing strategic adversaries offering bribes.


We simulate B-privacy scenarios using the same DAO dataset from Section~\ref{sec:pd_in_daos}. We first exclude voters whose choice was abstention, treating abstention as non-participation since it doesn't directly contribute to vote totals, doesn't represent a preference on the proposal outcome, and allows more proposals to be considered in our binary choice model. We then restrict to proposals with exactly two voting choices (due to our model's current limitations) and exclude proposals with more than 30,000 voters due to computational infeasibility of the iterative equilibrium computation required for B-privacy calculation. Applying these filters yields 3,582 proposals across 30 DAOs.

We compare three tally algorithms: full-disclosure $\disclosure_{\textsf{public}}$ (revealing individual votes), corrected noised tally $\disclosure_{\mathsf{noised}(\nu)+}$ of~\Cref{thm:noised-bribe-margin}, and winner-only $\disclosure_{\textsf{winner}}$ (revealing only the winning choice). Since we use an upper bound on bribe margins for the corrected noised tally (Theorem~\ref{thm:noised-bribe-margin}), our computed B-privacy values for this algorithm represent lower bounds on the true B-privacy, making our privacy improvements conservative estimates. 

All code and data used in our experiments are open source and available at \url{https://anonymous.4open.science/r/dao-voting-privacy-B65C}. 


\subsection{Experimental setup}

Our simulation requires modeling several components that are unobservable in practice, most notably voter utilities and adversarial objectives. We use the historical DAO voting data to calibrate realistic parameters for these components, then simulate how voters would behave under our B-privacy framework when facing strategic adversaries. To be consistent with our model, in which the adversary always attempts to maximize the probability of a $\mathsf{yes}$ outcome, we model the winning side of historical proposals as $\mathsf{no}$ and the losing side as $\mathsf{yes}$. This corresponds to an adversary that attempts to compromise an election by bribing voters to flip the outcome from the likely result based on voters' true utilities to the opposite outcome with high probability. 

\paragraph{Voter utilities.} We simulate voter utilities using normal distributions centered at the observed vote choice: $\mathbf{U}_i^{\mathsf{no}} \sim \mathcal{N}(1, 1)$ if voter $i$ voted $\mathsf{no}$ and $\mathbf{U}_i^{\mathsf{no}} \sim \mathcal{N}(-1, 1)$ if they voted $\mathsf{yes}$. This captures realistic uncertainty: a voter who voted $\mathsf{no}$ would vote $\mathsf{yes}$ in our model only if their sampled utility is negative, which occurs with probability $P(\mathcal{N}(1,1) < 0) \approx 16\%$. Thus roughly 84\% of voters stick with their observed choice, representing reasonable but imperfect knowledge that could be inferred from public information such as past voting patterns or public statements. Crucially, while varying $\sigma$ affects absolute B-privacy levels (higher uncertainty increases B-privacy), we find that relative B-privacy—the ratio between the value for given tally algorithm and the value for the winner-only tally algorithm—remains stable across different values of $\sigma$. This makes our comparative analysis robust to the specific modeling assumptions about voter utilities.


\paragraph{Noise distribution and tally perturbation.} To set noise for $\disclosure_{\text{noised}(\nu)+}$ We apply the tally perturbation framework from Section~\ref{sec:calibrating-noise} using Laplacian noise for 10\% tally perturbation with 95\% frequency. This choice balances privacy protection with transparency preservation. While we focus on Laplacian noise due to its established use in differential privacy, the choice of noise distribution may impact the privacy-transparency tradeoff and merits further exploration.

\paragraph{Adversarial success probability.} We model a highly motivated adversary seeking a high probability of success by setting the adversarial target success probability to $p = 0.9$. Results for other high values of $p$ are qualitatively similar.


\paragraph{Bribe allocation strategies.} As discussed in Section~\ref{sec:adv_strategy}, we simplify the complex coupled optimization problem by testing several reasonable bribe allocation strategies that target voters unlikely to support the adversary's preferred outcome, then computing the resulting equilibrium for each strategy and selecting the one yielding the lowest budget for each proposal. Appendix~\ref{app:comp-methods} provides full computational details, including the specific allocation strategies we tested.

\paragraph{Summary of parameters.} Table~\ref{tab:experimental_parameters} summarizes all simulation parameters and their justifications.

\begin{table*}[h]
\centering
\begin{tabular}{lll}
\toprule
\textbf{Parameter} & \textbf{Value/Distribution} & \textbf{Justification} \\
\midrule
Voter Utility Distribution & Normal distribution & Standard choice for modeling continuous preferences \\
\midrule
Utility Parameters & $U_i^{\text{no}} \sim \mathcal{N}(\mu_i, 1)$ & Models adversarial uncertainty about voter preferences; $\sigma = 1$ means \\
& $\mu_i = +1$ if $c_i = \mathsf{no}$ & voters stick with observed choice with $\approx$84\% probability\\
& $\mu_i = -1$ otherwise &  \\
\midrule
Noise Distribution & Laplace distribution & Natural choice from differential privacy literature; \\
& & other distributions merit exploration \\
\midrule
Noise Parameter / & $Y \sim \text{Lap}(0, 0.1W/\ln 20)$ & Ensures 95\% probability that noised tally differs from\\
Tally Perturbation & $\Pr(|Y| \leq 0.1W) = 0.95$ & true tally by $\leq 10\%$; balances privacy and transparency \\
\midrule
Adversarial Success & $p = 0.9$ & Models highly motivated adversary with strong success \\
Probability & & requirements \\
\bottomrule
\end{tabular}
\caption{Simulation parameters for empirical analysis of B-Privacy}
\label{tab:experimental_parameters}
\end{table*}

\subsection{Results}

Our empirical analysis examines B-privacy improvements across different tally algorithms and explores how proposal characteristics influence the privacy-transparency tradeoff. Before presenting the detailed findings, we introduce a key metric for interpreting results and provide an overview of our main conclusions.

\paragraph{Minimum Decisive Coalition (MDC).} To analyze how proposal characteristics affect B-privacy, we introduce the \emph{Minimum Decisive Coalition} (MDC). For a given proposal, the \emph{MDC} is the smallest number of voters who could \emph{change} their votes to flip the outcome. For instance, in~\Cref{ex:TP}, the MDC is 1, as either $\mathsf{yes}$ voter could deviate to change the winner from $\mathsf{yes}$ to $\mathsf{no}$. Two factors drive MDC: the presence of whales and the closeness of the proposal margin. Large token holders can single-handedly or in small coalitions determine outcomes, while tight margins make individual votes more pivotal—both resulting in low MDC values. However, in the DAO context, whale concentration is typically the dominant factor. For example in many proposals we analyze, a single voter controls >50\% of weight (yielding MDC = 1), this reflects genuine centralization regardless of the margin. Thus, while MDC captures multiple dynamics, low MDC values in our dataset primarily indicate voting weight concentration rather than competitive electoral outcomes.

\paragraph{Overview of results.} Our experiments reveal that corrected noised tallying calibrated to 10\% tally perturbation improves B-privacy across DAO proposals while preserving transparency, although this effect is broadly limited by low MDC. Across all 3582 proposals, corrected noised tally algorithms achieve a geometric mean relative B-privacy improvement of 1.5× over full-disclosure tallies, with winner-only algorithms achieving 2.9× improvement. However when we consider the 649 proposals with MDC$\geq 5$ these improvements increase to 4.1x and 14.2x respectively, demonstrating that even a modest MDC is sufficient to achieve substantial B-privacy. 




\begin{figure}[t!]
    \centering
    \includegraphics[width=1\linewidth]{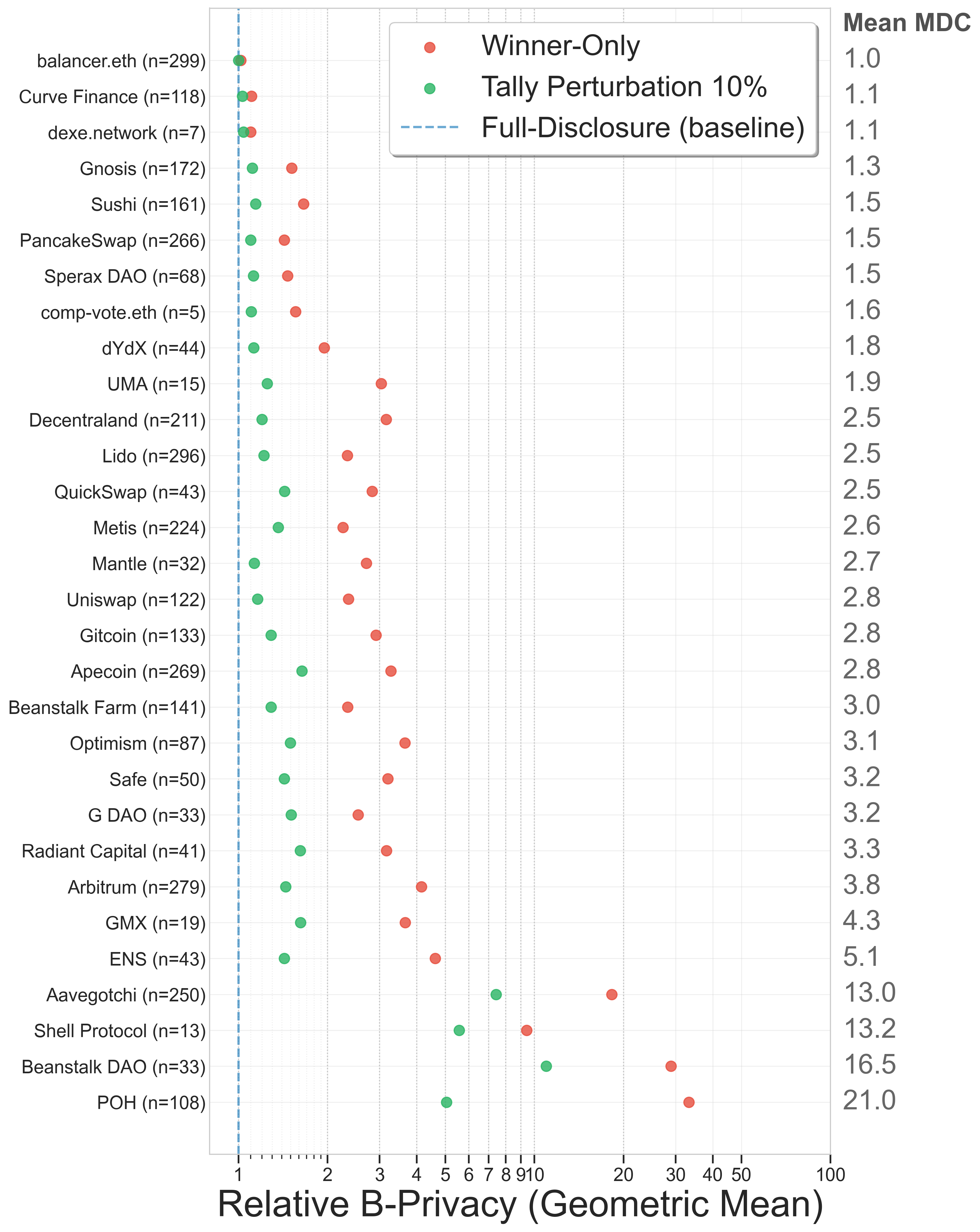}
    \caption{Relative B-Privacy by DAO under different Tally algorithms specified in Table \ref{tab:result-functions}. Results are averaged across proposals using geometric mean. Each row is a DAO, and rows are sorted using by the mean MDC across that DAOs proposals; red dots denote relative B-privacy in winner-only setting, green dots a lower bound on relative B-privacy in the corrected noised setting with tally perturbation $\talprtb=10\%$. The dashed vertical line marks the full-disclosure baseline. The $x$-axis is logarithmic; right-hand labels give the average MDC per DAO, left hand labels give the amount of proposals per DAO ($n$). In most DAOs using the corrected noised tally or winner-only algorithm improves B-Privacy, although the magnitude of this improvement is very dependent on MDC.}
    \label{fig:bribery_cost}
\end{figure}

Further exploring the simulation outcomes yields three broad results, which we expand upon here.

\paragraph{(B1) B-privacy across DAOs.} Figure~\ref{fig:bribery_cost} shows that across the proposals considered, the choice of tally algorithm does have an impact on B-privacy. B-privacy increases as the tally algorithm releases less information, with the winner-only algorithm providing maximal resistance to bribery and the corrected noised tally algorithm with a $10\%$ tally perturbation representing a tradeoff between bribery resistance and transparency. However, the fundamental finding is that \emph{most DAO proposals exhibit such low MDC (driven by large whales) that the choice of tally algorithm does not meaningfully affect B-privacy}. We see a clear delineation in relative B-privacy between 4 DAOs with average MDC $\geq13$ at the bottom of the plot (Proof of Humanity, Shell Protocol, Beanstalk DAO, and Aavegotchi) and the other DAOs which all have MDC $<6$. In the extreme case, almost all binary choice Balancer DAO proposals have a majority whale controlling $>50\%$ of voting weight, meaning any tally algorithm that reveals the outcome provides equivalent information to adversaries—rendering privacy mechanisms ineffective. More generally, when voting weight is heavily concentrated in low-MDC DAOs, even the winner-only tally algorithm provides only modest improvements: for many of these DAOs the relative B-privacy of the winner-only algorithm is $<3\times$, with corrected noised tallies offering little benefit over the public baseline. This highlights a fundamental limitation: tally privacy can only provide meaningful bribery resistance if voting weight is sufficiently decentralized.

\paragraph{(B2) Effect of noise calibration and MDC.} Despite limited effectiveness in most cases, Figure~\ref{fig:bribery_cost} does indicate that the corrected noised tally mechanism can meaningfully improve B-privacy in proposals with even a modest MDC. We explore this further in Figure~\ref{fig:B_privacy_vs_noise}, which plots relative B-Privacy versus tally perturbation for the Aavegotchi DAO, grouping proposals by their MDC. We focus on Aavegotchi because it offers both a wide range of MDC values and a large number of proposals (\(n=250\)), making the trends especially clear. Aggregating within MDC cohorts reveals two clear patterns: (i) B-Privacy rises monotonically with tally perturbation; and (ii) for any fixed tally perturbation, proposals with larger MDC—i.e., requiring larger coalitions to flip outcomes—achieve higher B-Privacy. This pattern holds broadly across DAOs: adding more noise consistently raises B-Privacy (at the cost of transparency), while the effectiveness of the mechanism is fundamentally constrained by centralization, as measured by MDC. B-Privacy rises steeply with small amounts of noise but quickly plateaus, so adding more than about 10\% tally perturbation offers little additional benefit in practice. As discussed previously, most DAOs exhibit very low average MDC, which explains why relative B-Privacy remains limited even under high noise or winner-only tallying. 

\paragraph{(B3) Optimal adversarial strategy.} Figure~\ref{fig:Bribe_Distribution} shows bribe allocation across voter weights for different tally algorithms in a representative proposal, illustrating that an adversary's most effective strategy targets voters whose compliance can be verified with highest fidelity. Under the corrected noised tally, Theorem~\ref{thm:noised-bribe-margin} shows smaller voters are less attractive both due to their lower pivotality and because when voters are not pivotal, adversaries must rely on distinguishing between noise distributions to infer votes. For smaller voters, their lower weight $w_i$ means these noise distributions overlap more substantially, making it harder for adversaries to determine which choice they made. We confirm this empirically: even under full-disclosure, only the largest voters are worth bribing—in our representative ApeCoin proposal\footnote{Proposal id: \texttt{\label{ftn:proposal_id}0x5b495182b087481490a79891cfd6456ea05473451a7\\a47b0f73f306ea8c5ee40}} with 452 voters, only 29 receive bribes, even when the full-disclosure tally algorithm is used. As tally algorithms disclose less information (moving from full-disclosure to corrected noised to winner-only), this number decreases further as more voters become effectively hidden by noise. Other proposals exhibit the same pattern. 

\begin{keyFinding}[title={Key Findings}]
\begin{enumerate}[label=\textbf{(B\arabic*)}]
  \item \textbf{Tally privacy provides minimal B-Privacy in most DAOs due to large whale presence.}
  \item \textbf{More noise consistently increases B-privacy, with effectiveness increasing at higher MDC.}
  \item \textbf{Adversaries optimally target whales regardless of tally algorithm.}
\end{enumerate}
\end{keyFinding}

\begin{figure}[tp]
    \centering
    \includegraphics[width=1\linewidth]{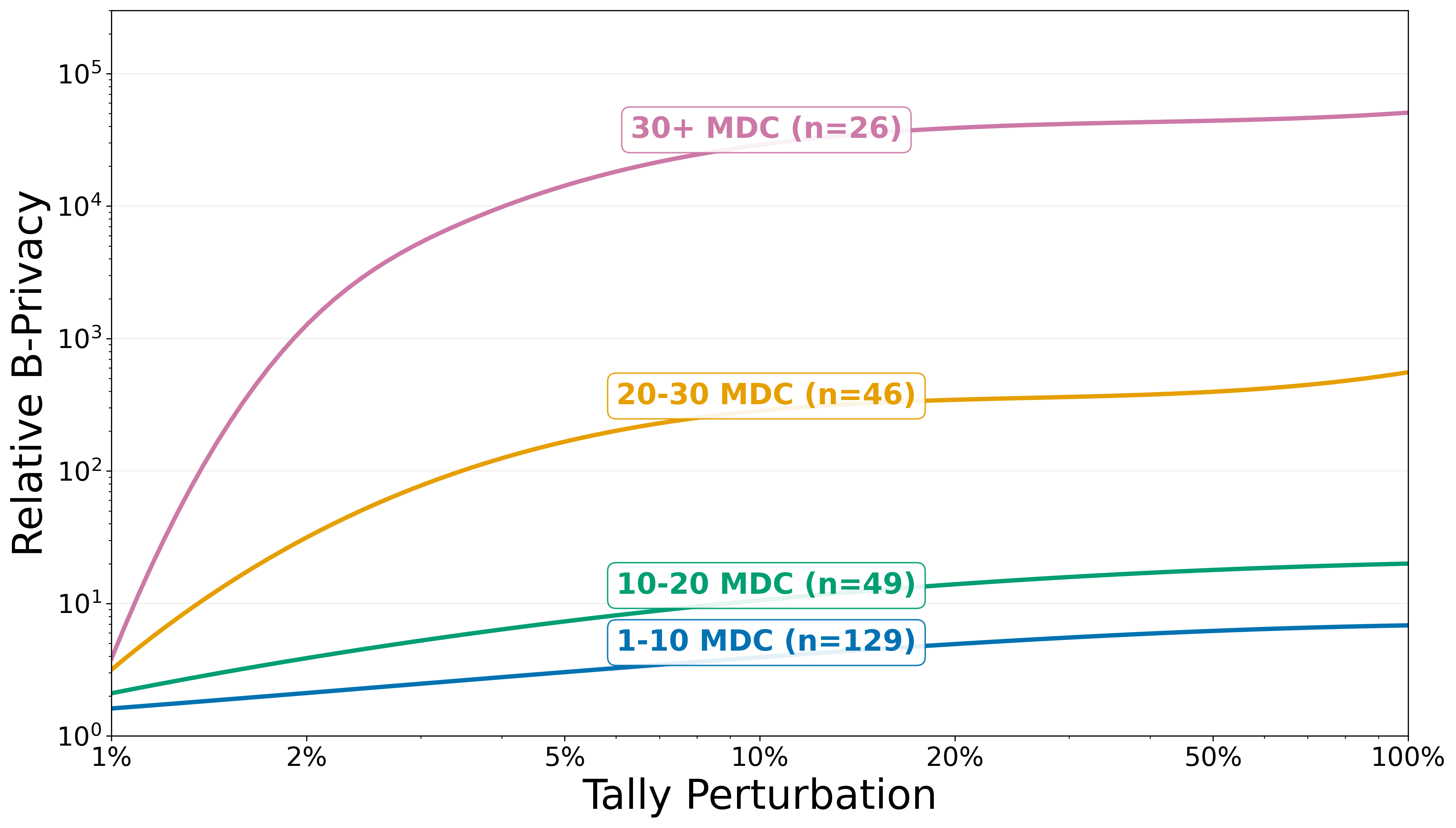}
    \caption{Relative B-Privacy as a function of tally perturbation $\talprtb$, grouped by minimum decisive coalition (MDC) size across 250 proposals from Aavegotchi~\cite{aavegotchi_website}. Each line averages the relative B-Privacy across all proposals within an MDC cohort ($n$ denotes the number of such proposals). The y-axis shows relative B-Privacy on a log scale. Relative B-Privacy rises with $\talprtb$. At any fixed $\talprtb$ larger MDC cohorts achieve higher relative B-Privacy.}
    \label{fig:B_privacy_vs_noise}
\end{figure}

\begin{figure}[tp]
    \centering
    \includegraphics[width=1\linewidth]{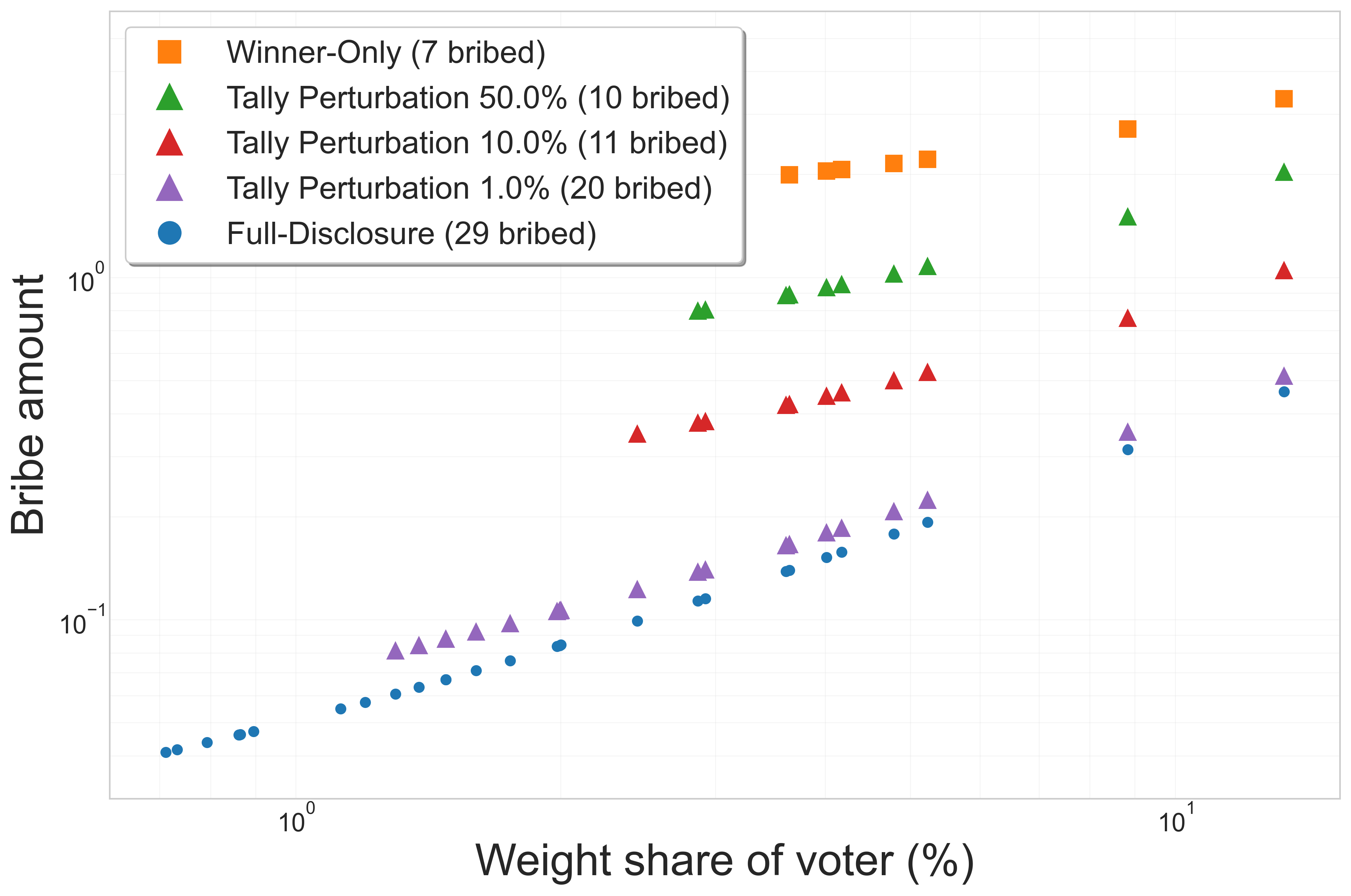}
    \caption{Optimal bribe distribution across voters for winner-only, corrected noised, and full-disclosure tally algorithms. Plotted on representative proposal \texttt{0x5b495182b08...} from ApeCoin with 452 voters and an MDC of 6. X-axis corresponds to weight shares, i.e., fraction of total weight, of individual voters. Y-axis shows bribe amount received by individual voters, measured in utility units. Where there is no dot on Y-axis for a given tally algorithm, the corresponding voter receives no bribe. As a tally algorithm discloses less information, the adversary bribes fewer voters.}
    \label{fig:Bribe_Distribution}
\end{figure}
\subsection{Summary guidance for DAOs}

In weighted voting, privacy is more than a voter-centric right—it is a structural defense against economic manipulation. Even when ballots remain hidden, precise results can leak enough information for adversaries to mount highly efficient bribery attacks. This undermines both decision integrity and community trust.

Drawing from our empirical evaluation, several practices emerge as effective in boosting B-Privacy without unduly degrading transparency:
\begin{itemize}
    \item \textbf{Apply (correct) noised tally algorithms:} Even minimal Laplacian noise, tuned to keep results within 10\% of the exact tally for most proposals, can raise B-Privacy, however the impact varies widely depending on the weight distribution and proposal dynamics.
    \item \textbf{Adjust tally perturbation for proposal sensitivity:} Apply stronger tally perturbations for high-value or contentious proposals where privacy is at a premium, less perturbation for routine decisions.
    \item \textbf{Account for weight distribution:} Skewed token holdings increase privacy risks. Our findings show that whale dominance is an even bigger threat to B-Privacy than the choice of tally algorithm. Reducing presence of whales can increase MDC across proposals such that small tally perturbations cause sharp increases in B-privacy.
    \item \textbf{Focus anti-bribery measures on whales:} Since adversaries optimally target high-weight voters who offer the best bribe margins, governance systems should prioritize monitoring and protecting whale voters. Traditional privacy mechanisms primarily shield smaller voters, but whales remain the most attractive and vulnerable targets for manipulation regardless of noise levels.
\end{itemize}

In our experiments, modest noise preserved the key benefits of transparency—e.g., enabling voters to understand the magnitude of winning margins—while increasing the adversary’s bribery cost. The results suggest that strong B-Privacy and transparency are not inherently in conflict; in the right circumstances and with careful tuning, both can be achieved in practice.

\section{Conclusion}
\label{sec:conclusion}

We have introduced B-Privacy, a new metric for privacy in the weighted-voting setting, where classical voting-privacy notions such as ballot secrecy are insufficient. B-Privacy measures the cost of bribery to an adversary induced by different choices of tally algorithm. It offers an economic lens on a key consequence of privacy loss in weighted voting: as adversaries gain more precise knowledge of voter behavior, their cost of bribery decreases, raising systemic security risks. 

Our work gives rise to a number of future research directions, among them: 

\begin{itemize}

\item \textbf{Multi-choice proposals:} 
As initial work, our B-privacy framework currently supports only binary-choice proposals. Extending to abstention (as an explicit ballot choice) raises new issues around participation cost and community dynamics. More generally, multiple ballot choices require modeling far more complex strategic coordination and bribery schemes.


\item \textbf{Analytical equilibrium characterization:} Our current approach relies on computational methods to find Bayesian Nash equilibria. Analytic bounds would provide deeper theoretical insights.

\item \textbf{Alternative governance mechanisms:} Our experimental results show how large whale presence (low MDC) in DAOs today raises the risk of bribery. Our work injects new urgency into the question of how DAOs can protect against whale dominance---a popular topic of study and community action~\cite{wang2025balancingsecurityliquiditytimeweighted,sharma2024unpacking}.


\item \textbf{Parameter exploration and modeling assumptions:} Our experiments make several simplifying assumptions that merit further investigation. We model voter utilities using identical normal distributions across voters and use Laplacian noise for tally perturbations. Future work could explore more realistic utility models that capture heterogeneity in voter preferences, investigate how different noise distributions (beyond Laplacian) affect the privacy-transparency tradeoff, and more generally examine the robustness of our results to alternative modeling choices and simulation parameters.
\end{itemize}

In summary, B-Privacy yields both theoretical results and practical guidance for DAO communities and other weighted-voting settings. With the growing reliance on weighted voting in blockchain governance and a movement toward secret-ballot systems, as well as the mounting privacy-related threats to system integrity, B-Privacy promises to serve as serve a key metric for evaluating privacy.

\paragraph{Acknowledgments.}
This work was funded by NSF CNS-2112751, generous support from IC3 industry partners and sponsors, Oasis Labs, Optimism RetroPGF, and a Sui Academic Research Award. Disclosure: Ari Juels is Chief Scientist at Chainlink Labs.




\ifACM \bibliographystyle{ACM-Reference-Format.bst} \fi
\ifUSENIX \bibliographystyle{plain} \fi
\ifIEEE \bibliographystyle{plain} \fi
\ifLNCS \bibliographystyle{Conferences/LNCS/splncs04.bst} \fi

\iffull \else \bibliography{references} \fi

\newcommand{\mpk}{\textsf{mpk}}
\newcommand{\msk}{\textsf{msk}}
\newcommand{\getsr}{{{\leftarrow{\hspace*{-3pt}\raisebox{.75pt}{$\scriptscriptstyle\$$}}}}}
\newcommand{\calM}{\mathcal{M}}

\appendix

\section{Additional Theorems and Proofs}\label{app:optimal_strategy}

This appendix contains the formal proofs of theorems presented in Sections~\ref{sec:B-privacy_comp} and \ref{sec:noise_mechanism} which cover our analytic results for computing B-privacy and bounding B-privacy for noised tally algorithms.

\subsection{Independence of bribery condition functions} \label{app:optimal_strategy_indep}

\begin{theorem}[Independence of bribery condition functions]\label{thm:independence}
For any bribery condition function $f: O \to \{0,1\}^n$ (possibly randomized and correlated across voters), there exists an equivalent independent strategy using condition functions $f_i: O \to \{0,1\}$ that achieves the same bribe margins $\bm_i$ for all voters $i$, and hence the same B-privacy.
\end{theorem}
\begin{proof}
Given any (possibly randomized, correlated) condition function $f$, construct independent functions by setting:
$$f_i(o) = f(o)_i \text{ for each voter } i \text{ and outcome } o.$$

Crucially, each function $f_i$ represents an independent invocation of the original correlated function $f$—not a single correlated invocation whose components are distributed across voters. This eliminates correlations while preserving the marginal payment probability for each voter.

The bribe margin for voter $i$ under either strategy is:
\begin{align*}
&\Pr[f(\disclosure(t))_i = 1 \mid c_i = \text{yes}] - \Pr[f(\disclosure(t))_i = 1 \mid c_i = \text{no}] \\
&= \Pr[f_i(\disclosure(t)) = 1 \mid c_i = \text{yes}] - \Pr[f_i(\disclosure(t)) = 1 \mid c_i = \text{no}].
\end{align*}

Since each voter's decision depends only on their individual expected payoff—which depends only on their own bribe margin and pivotality (Theorem~\ref{thm:adv_success_prob})—the transformation preserves equilibrium behavior and adversarial success probability.
\end{proof}

\subsection{Proof of optimal bribery condition functions}\label{app:optimal_strategy_bcf}

For clarity we prove the first part of Theorem \ref{thm:optimal-bcf} as a lemma:

\begin{lemma}[Optimality of maximal bribe margin]\label{lemma:advantage}
For the adversary in the bribery game, the optimal bribery condition functions are those that maximize the bribe margin $\bm_i$ for each voter $i$.
\end{lemma}
\begin{proof}
Suppose for contradiction that there exists an optimal solution $(\mathbf{b}^*, \mathbf{f}^*)$ with some condition function $f_k^*$ that does not maximize voter $k$'s bribe margin. Let $\bm_k$ denote the bribe margin of $f_k^*$, and let $\bm_k^{\max} > \bm_k$ be the maximal achievable bribe margin for voter $k$.

Since $\bm_k^{\max} > \bm_k$, we can write $\bm_k^{\max} = a \bm_k$ for some $a > 1$. 

Consider an alternative strategy $(\mathbf{b}', \mathbf{f}')$ where:
\begin{itemize}
    \item $f_i' = f_i^*$ for all $i \neq k$
    \item $f_k'$ achieves the maximal bribe margin $\bm_k^{\max}$
    \item $b_i' = b_i^*$ for all $i \neq k$
    \item $b_k' = \frac{b_k^*}{a}$
\end{itemize}

By Theorem \ref{thm:adv_success_prob}, voter $k$'s expected payoff from accepting the bribe is:
$$\bm_k \cdot b_k^* = a \bm_k \cdot \frac{b_k^*}{a} = \bm_k^{\max} \cdot b_k'.$$

Since the expected payoff is unchanged, voter $k$'s equilibrium behavior remains the same. The behavior of all other voters is also unchanged, so the success probability remains $p$.

However, the total bribe budget decreases:
$$\sum_{i=1}^n b_i' = \sum_{i \neq k} b_i^* + b_k' < \sum_{i \neq k} b_i^* + b_k^* = \sum_{i=1}^n b_i^*.$$

This contradicts the assumption that $(\mathbf{b}^*, \mathbf{f}^*)$ was optimal.
\end{proof}

We now proceed to considering the form of the optimal bribery condition functions:

\begin{proof}[Proof of Theorem \ref{thm:optimal-bcf}]
By Lemma \ref{lemma:advantage}, the optimal condition function for voter $i$ is the one that maximizes the bribe margin $\bm_i$. We now characterize this function.

Define $p_{i,c}^{o} = \Pr[\disclosure(t) = o \mid c_i = c]$. For any condition function $f_i: O \to \{0,1\}$, consider the definition of the bribe margin:
\begin{align*}
\bm_i &= \Pr[f_i(\disclosure(t)) = 1 \mid c_i = \mathsf{yes}] - \Pr[f_i(\disclosure(t)) = 1 \mid c_i = \mathsf{no}] \\
&= \sum_{o \in O} f_i(o) p_{i,\mathsf{yes}}^{o} - \sum_{o \in O} f_i(o) p_{i,\mathsf{no}}^{o} \\
&= \sum_{o \in O} f_i(o) \left(p_{i,\mathsf{yes}}^{o} - p_{i,\mathsf{no}}^{o}\right).
\end{align*}

Since $f_i(o) \in \{0,1\}$, this sum is maximized by setting $f_i(o) = 1$ if and only if the term in parentheses is non-negative, which gives:
$$f^*_i(o) = \mathbb{I}\{p_{i,\mathsf{yes}}^{o} \geq p_{i,\mathsf{no}}^{o}\}.$$

This yields the optimal bribe margin:
$$\bm_i^* = \sum_{o \in O} \max(p_{i,\mathsf{yes}}^{o} - p_{i,\mathsf{no}}^{o}, 0).$$
\end{proof}

\newcounter{savedtheorem}
\setcounter{savedtheorem}{\value{theorem}}
\setcounter{theorem}{7}
\begin{corollary}\label{crly:std_cases}
The optimal bribery condition functions for standard tally algorithms are (where $o=\disclosure(t)$ denotes the outcome):
\begin{itemize}
    \item $\disclosure_{\mathsf{winner}}$: $f^*_i(o) = \mathbb{I}\{o = \mathsf{yes}\}$ with bribe margin $\bm^*_i=\Delta_i$
    
    \item $\disclosure_{\mathsf{public}}$: $f^*_i(o) = \mathbb{I}\{o_i = \mathsf{yes}\}$ with bribe margin $\bm^*_i =1$
\end{itemize}
\end{corollary}
\setcounter{theorem}{\value{savedtheorem}}
\begin{proof}
We apply the optimal condition function form to each tally algorithm:

\medskip\noindent\textbf{(1) Winner-only algorithm.}
Here $\disclosure(t) \in \{\mathsf{yes},\mathsf{no}\}$ reveals only the winning choice. Since voting $\mathsf{yes}$ cannot make a $\mathsf{yes}$ outcome less likely, we have:
\[
\Pr[\disclosure(t)=\mathsf{yes} \mid c_i=\mathsf{yes}] \geq \Pr[\disclosure(t)=\mathsf{yes} \mid c_i=\mathsf{no}].
\]

By Theorem \ref{thm:optimal-bcf}, the optimal condition function is $f_i(o) = \mathbb{I}\{o = \mathsf{yes}\}$.

The bribe margin is:
\begin{align*}
\bm_i &= \Pr[\mathsf{yes} \text{ wins} \mid c_i = \mathsf{yes}] - \Pr[\mathsf{yes} \text{ wins} \mid c_i = \mathsf{no}] \\
&= \Pr[\mathsf{no} \text{ wins} \mid c_i = \mathsf{no}] - \Pr[\mathsf{no} \text{ wins} \mid c_i = \mathsf{yes}] \\
&= \Delta_i.
\end{align*}

\medskip\noindent\textbf{(2) Full-disclosure algorithm.}
Here $\disclosure(t) \in \{\mathsf{yes},\mathsf{no}\}^n$ reveals all individual votes. Since $\disclosure(t)_i = c_i$ always:
\[
\begin{aligned}
\Pr[\disclosure(t)_i=\mathsf{yes} \mid c_i=\mathsf{yes}] &= 1 \\ \geq \Pr[\disclosure(t)_i=\mathsf{yes} \mid c_i=\mathsf{no}] &= 0.
\end{aligned}
\]

By Theorem \ref{thm:optimal-bcf}, the optimal condition function is $f_i(o) = \mathbb{I}\{o_i = \mathsf{yes}\}$ with bribe margin $\bm_i = 1$.
\end{proof}

\subsection{Bound on Plausible Deniability}\label{app:bound_PD}

\begin{proof}[Proof of Theorem \ref{lem:bribe-margin-epd}]

 Let $p_{i,c}^{o} = \Pr[\disclosure(t)=o| c_i = c]$ and $p_{o}^{i,c} = \Pr[c_i = c| \disclosure(t)=o]$. Using Bayes' rule we have
    \[
        p_{i,c}^{o} = \frac{p^{i,c}_{o}\Pr[\disclosure(t)=o]}{\Pr[c_i=c]}.
    \]

    We also have that
    \[
        \max(p_{i,\yes}^o - p_{i,\no}^o, 0) + \min(p_{i,\yes}^o, p_{i,\no}^o) = p_{i,\yes}^o.
    \]
    
    Then we can finally claim
    \begin{align*}
    \bm_i^*
        &= \sum_{o\in O} \max(p_{i,\yes}^o - p_{i,\no}^o, 0) \\
        &= \sum_{o\in O} p_{i,\yes}^o - \sum_{o\in O} \min\left(p_{i,\yes}^o, p_{i,\no}^o\right)\\
        &=\sum_{o\in O} \frac{p^{i,\yes}_{o}}{\Pr[c_i=\yes]}\Pr[\disclosure(t)=o] \\&\quad - \sum_{o\in O} \min\left(\frac{p^{i,\yes}_{o}}{\Pr[c_i=\yes]},
              \frac{p^{i,\no}_{o}}{\Pr[c_i=\no]}\right)\Pr[\disclosure(t)=o]\\
        &=1 - EPD_i^{\disclosure}.        
    \end{align*}
\end{proof}

\subsection{Differentially private tally algorithm}
\label{app:dp-disclosure}

Consider the following noised tally algorithm for binary proposals that uses Laplace noise. We write $Lap(b)$ to denote the Laplace distribution with scale $b$ and use $w_{\max}=\max_i w_i$:
\[
\disclosure_{\mathsf{noised}\left(Lap(w_{max}/\cmepsilon)\right)} = Y+\sum_{i: c_i = \mathsf{yes}} w_i, \quad Y \stackrel{\$}{\leftarrow}  Lap(w_{max}/\cmepsilon).
\]

For brevity, we denote this tally algorithm as $\disclosure_{\mathsf{dp}}$.

\begin{theorem}
    When the noised tally algorithm $\disclosure_{\mathsf{dp}}$ is used the optimal bribery condition function for any voter has bribe margin at most $1-e^{-\cmepsilon}$.
\end{theorem}

\begin{proof}
    
    For any two adjacent voting transcripts $t,t'$ (differing in one voter's choice) we calculate the maximum possible difference in $\sum_{i: c_i = \mathsf{yes}} w_i$. We have $\max|\sum_{i: c_i = \mathsf{yes}} w_i-\sum_{i: c'_i = \mathsf{yes}} w_i|=w_{\max}$ so the $\ell_1$ sensitivity of this sum is $w_{\max}$. Then the noised tally algorithm $\disclosure_{\mathsf{dp}}$ is clearly the Laplace mechanism applied to this sum, which is $\cmepsilon$-differentially private \cite{dwork2006calibrating}.

    By the definition of differential privacy, for any outcome $o$ and adjacent transcripts $t,t'$ differing only in voter $i$'s choice:

    \[
    \frac{\Pr\left[\disclosure_{\mathsf{dp}}(t)=o\right]}{\Pr\left[\disclosure_{\mathsf{dp}}(t')=o\right]} \leq e^\cmepsilon.
    \]
    
    Since adjacent transcripts differing in voter $i$'s choice correspond exactly to $c_i = \mathsf{yes}$ versus $c_i = \mathsf{no}$, this gives us:
    
    \[
    \frac{p_{i,\mathsf{yes}}^{o}}{p_{i,\mathsf{no}}^{o}} = \frac{\Pr[\disclosure(t) = o \mid c_i = \mathsf{yes}]}{\Pr[\disclosure(t) = o \mid c_i = \mathsf{no}]} \leq e^\cmepsilon.
    \]
    
    By the post-processing property of differential privacy, applying any function (including the condition function $f_i$) cannot increase the privacy loss, so:
    
    \[
    \frac{\Pr\left[f_i(\disclosure(t))=1|c_i=\mathsf{yes}\right]}{\Pr\left[f_i(\disclosure(t))=1|c_i=\mathsf{no}\right]} \leq e^\cmepsilon.
    \]
    
   Let $P_{\mathsf{yes}} = \Pr[f_i(\disclosure(t))=1|c_i=\mathsf{yes}]$ and $P_{\mathsf{no}} = \Pr[f_i(\disclosure(t))=1|c_i=\mathsf{no}]$. Since these are probabilities, we have $P_{\mathsf{yes}} = \min(e^\cmepsilon P_{\mathsf{no}}, 1)$. The bribe margin is $\bm_i = P_{\mathsf{yes}} - P_{\mathsf{no}}$. We consider two cases:

    Case 1: $e^\cmepsilon P_{\mathsf{no}} < 1$, so $P_{\mathsf{no}} < e^{-\cmepsilon}$:
    \begin{align*}
    \bm_i &= P_{\mathsf{yes}} - P_{\mathsf{no}} \leq e^{\cmepsilon}P_{\mathsf{no}} - P_{\mathsf{no}}\\
    &= P_{\mathsf{no}}(e^\cmepsilon-1) \leq e^{-\cmepsilon}(e^\cmepsilon-1) = 1-e^{-\cmepsilon}.
    \end{align*}
    
    Case 2: $e^\cmepsilon P_{\mathsf{no}} \geq 1$, so $P_{\mathsf{no}} \geq e^{-\cmepsilon}$:
    \begin{align*}
    \bm_i &= P_{\mathsf{yes}} - P_{\mathsf{no}} \leq 1 - P_{\mathsf{no}} \leq 1- e^{-\cmepsilon}.
    \end{align*}
    
    Therefore the maximum bribe margin is $1-e^{-\cmepsilon}$.
\end{proof}

\subsection{Bound on corrected noised tally bribe margin}\label{app:optimal_strategy_noisy}

\begin{proof}[Proof of Theorem \ref{thm:noised-bribe-margin}]

    We consider the tally algorithm $\disclosure_{\mathsf{noised}(\nu)+}(t)=(\disclosure_{\mathsf{noised}(\nu)}(t), \disclosure_{\mathsf{winner}}(t)).$

        Throughout this proof we abbreviate notation, writing $\disclosure_{\mathsf{n+}}$ for $\disclosure_{\mathsf{noised}(\nu)+}$, $\disclosure_{\mathsf{n}}$ for $\disclosure_{\mathsf{noised}(\nu)}$, $\disclosure_{\mathsf{w}}$ for $\disclosure_{\mathsf{winner}}$, and $\Pr[\disclosure=o \mid c]$ for $\Pr[\disclosure(t) = o \mid c_i = c]$.
    
    By Theorem \ref{thm:optimal-bcf}, for any voter $i$ when using the optimal bribery condition function, the bribe margin is:
    \begin{align*}
     \bm^*_i&=\sum_{o \in O}\max\left(\Pr\left[\disclosure_{\mathsf{n+}}=o \mid \mathsf{yes}\right] - \Pr\left[\disclosure_{\mathsf{n+}}=o \mid \mathsf{no}\right], 0\right).
    \end{align*}
    
    Since the combined outcome pairs $(o_1, o_2)$ where $o_1$ is the noised tally and $o_2 \in \{\mathsf{yes}, \mathsf{no}\}$ is the winner, we decompose the sum over all possible outcomes:
    \begin{align*}
     \bm^*_i &= \sum_{o_2\in\{\mathsf{yes}, \mathsf{no}\}}\int_{-\infty}^{\infty}\max(\Pr\left[\disclosure_{\mathsf{n}}=o_1, \disclosure_{\mathsf{w}}=o_2\mid \mathsf{yes}\right]\\&-\Pr\left[\disclosure_{\mathsf{n}}=o_1, \disclosure_{\mathsf{w}}=o_2\mid \mathsf{no}\right],0)\,do_1.
    \end{align*}

    Let $\mathbb{T}$ be the set of all possible sums of weights of voters that choose $\mathsf{yes}$ excluding voter $i$, and $T$ be the random variable for this value with probability taken over all other voters' choices. Let $Z \sim \nu$ be the noise random variable. Define the winner function:
    $$\mathsf{winner}(x)=\begin{cases}
        \mathsf{yes} & \text{if } x \ge \frac{W}{2}\\
        \mathsf{no} & \text{otherwise.}
    \end{cases}$$
    
    We condition on the partial tally $T$ to decompose each probability. When voter $i$ votes $\mathsf{yes}$, the total $\mathsf{yes}$ tally is $s + w_i$; when voting $\mathsf{no}$, it is $s$:
    \begin{align*}
        &\Pr\left[\disclosure_{\mathsf{n}}=o_1, \disclosure_{\mathsf{w}}=o_2\mid\mathsf{yes}\right] \\
        &=\sum_{s \in \mathbb{T}}\Pr[T=s]\Pr\left[Z=o_1-(s+w_i)\right] \cdot \mathbb{I}\left\{o_2 =\mathsf{winner}(s+w_i)\right\}.\\\\
        &\Pr\left[\disclosure_{\mathsf{n}}=o_1, \disclosure_{\mathsf{w}}=o_2\mid\mathsf{no}\right] \\
        &=\sum_{s \in \mathbb{T}}\Pr[T=s]\Pr\left[Z=o_1-s\right] \cdot \mathbb{I}\left\{o_2 =\mathsf{winner}(s)\right\}.
    \end{align*}
    
    For clarity, define:
    \begin{align*}
    P_y(s, o_1, o_2) &= \Pr[Z=o_1-(s+w_i)] \cdot \mathbb{I}\{o_2 = \mathsf{winner}(s+w_i)\}, \\
    P_n(s, o_1, o_2) &= \Pr[Z=o_1-s] \cdot \mathbb{I}\{o_2 = \mathsf{winner}(s)\}.
    \end{align*}

    For brevity, we omit the arguments $(s, o_1, o_2)$ when the context is clear.
    
    Substituting back into the bribe margin calculation and applying the inequality $\max\left(\sum_{s}f(s),0\right) \leq \sum_{s}\max\left(f(s), 0\right)$:
    \begin{align*}
        \bm^*_i&= \sum_{o_2\in\{\mathsf{yes}, \mathsf{no}\}}\int_{-\infty}^{\infty}\max \left(\sum_{s \in \mathbb{T}}(P_y-P_n)\Pr\left[T=s\right], 0\right)\,do_1 \\
        & \leq \sum_{o_2\in\{\mathsf{yes}, \mathsf{no}\}}\int_{-\infty}^{\infty} \sum_{s \in \mathbb{T}} \max\left((P_y-P_n)\Pr\left[T=s\right], 0\right)\,do_1 \\
        &= \sum_{o_2\in\{\mathsf{yes}, \mathsf{no}\}}\int_{-\infty}^{\infty} \sum_{s \in \mathbb{T}} \max \left(P_y-P_n, 0\right)\Pr\left[T=s\right]\,do_1 \\
        &= \sum_{s \in \mathbb{T}}\Pr\left[T=s\right]\sum_{o_2\in\{\mathsf{yes}, \mathsf{no}\}}\int_{-\infty}^{\infty} \max \left(P_y-P_n, 0\right)\,do_1. 
    \end{align*}

    We now decompose the sum over partial tallies based on whether voter $i$ is pivotal.
    
    First consider the partial tallies in which voter $i$ is pivotal. Let $\mathbb{T}^{\mathsf{pivotal}} = \{s \in \mathbb{T} : \frac{W}{2}-w_i < s < \frac{W}{2}\}$. For $s \in \mathbb{T}^{\mathsf{pivotal}}$, voter $i$'s choice determines the winner: $\mathsf{winner}(s+w_i) \neq \mathsf{winner}(s)$. This means that for any fixed $o_2$, exactly one of the indicators $\mathbb{I}\{o_2 = \mathsf{winner}(s+w_i)\}$ or $\mathbb{I}\{o_2 = \mathsf{winner}(s)\}$ equals 1, so exactly one of $P_y$ and $P_n$ is nonzero. Therefore $\max(P_y-P_n,0) = P_y$, and we have:
    \begin{align*}
     & \sum_{s \in \mathbb{T}^{\mathsf{pivotal}}} \Pr[T=s]\sum_{o_2\in\{\mathsf{yes}, \mathsf{no}\}}\int_{-\infty}^{\infty} \max (P_y-P_n, 0)\,do_1 \\ 
     &= \sum_{s \in \mathbb{T}^{\mathsf{pivotal}}} \Pr[T=s]\sum_{o_2\in\{\mathsf{yes}, \mathsf{no}\}}\int_{-\infty}^{\infty}P_y\,do_1 \\
     &= \sum_{s \in \mathbb{T}^{\mathsf{pivotal}}} \Pr[T=s] \\
     &= \Pr[\text{voter $i$ is pivotal}] = \Delta_i.
    \end{align*}

    Now consider the partial tallies in which voter $i$ is not pivotal. In these cases, $\mathsf{winner}(s+w_i) = \mathsf{winner}(s)$, so for any fixed $o_2$, the indicators $\mathbb{I}\{o_2 = \mathsf{winner}(s+w_i)\}$ and $\mathbb{I}\{o_2 = \mathsf{winner}(s)\}$ are either both 1 or both 0. We need only consider the case where both equal 1, allowing us to drop the sum over $o_2$:
    
    Let $q(u) = \Pr[Z=u-w_i] - \Pr[Z=u]$. Then:
    \begin{align*}
    & \sum_{s \in \mathbb{T} \setminus \mathbb{T}^{\mathsf{pivotal}}}\Pr[T=s]\sum_{o_2\in\{\mathsf{yes}, \mathsf{no}\}}\int_{-\infty}^{\infty} \max (P_y-P_n, 0)\,do_1 \\
    &= \sum_{s \in \mathbb{T} \setminus \mathbb{T}^{\mathsf{pivotal}}}\Pr[T=s]\int_{-\infty}^{\infty} \max (q(o_1-s), 0)\,do_1 \\
    &= \sum_{s \in \mathbb{T} \setminus \mathbb{T}^{\mathsf{pivotal}}}\Pr[T=s]\int_{-\infty}^{\infty} \max (q(u), 0)\,du \\
    &= (1-\Delta_i)\int_{-\infty}^{\infty} \max (q(u), 0)\,du.
    \end{align*}

    Combining both sets of partial tallies and substituting back $q(u) = \Pr[Z=u-w_i] - \Pr[Z=u]$ completes the proof:
    \begin{align*}
        \bm^*_i &\leq \Delta_i + (1-\Delta_i)\int_{-\infty}^{\infty} \max\left(\Pr[Z=u-w_i] - \Pr[Z=u], 0 \right)du. 
    \end{align*}
\end{proof}

\begin{table*}[h]
\centering
\begin{tabular}{ll}
\toprule
\textbf{Bribery Allocation Strategy} & \textbf{Bribe Distribution Formula} \\
\midrule
Equal Split & $b_i = B/n$ \\
\midrule
Linear & $b_i \propto w_i$ \\
\midrule
Square-Root & $b_i \propto \sqrt{w_i}$ \\
\midrule
Quadratic & $b_i \propto w_i^2$ \\
\midrule
Logarithmic & $b_i \propto \log(w_i)$ \\
\midrule
Linear Sloped$(s)$ & $b_i = sw_i + c$ \\
\bottomrule
\end{tabular}
\caption{Bribe allocation strategies considered. All strategies target only voters opposing the adversary's preferred outcome, with bribes normalized to satisfy budget constraint $\sum_i b_i = B$. In the bribe distribution formulas $n$ represents the number of voters, $s$ represents the slope and $c$ is chosen so that the allocation satisfies the budget constraint.}
\label{tab:experimental_alloc_func}
\end{table*}

\section{Computational Methods for B-Privacy}
\label{app:comp-methods}


This appendix details the methods used to compute B-privacy values $B_{\disclosure}(p)$ for a tally algorithm $\disclosure$ at target success probability $p$. As discussed in Section~\ref{sec:adv_strategy}, this requires solving two interconnected problems: determining the Bayesian Nash equilibrium and finding optimal bribe allocation. We address the first problem using fixed-point iteration to approximate equilibrium behavior, and the second by testing reasonable heuristic allocation strategies rather than attempting to solve the complex coupled optimization problem.

For each proposal, the procedure returns the minimal total bribe $B^*$ and an associated per-voter bribe vector $\mathbf{b}^*$ such that $p_{{\sf succ}} \geq p$ 

\paragraph{Inputs and utility model.}
Let $\mathbf{w}=(w_1,\cdots,w_n)$ be voter weights, $W=\sum_i w_i$ the total voting weight and $\mathbf{c} = (c_1, \cdots, c_n)$ the observed voter choices on a given proposal.

For computational tractability, in our experiments, we model each voter $i$'s utility $U_i^{\sf no}$ as Gaussian, although other distributions could be used. We set the prior based on their observed vote: voters who voted for the winning choice are modeled as having higher utility for ``no'' ($\mu_i = +1$), while voters who voted for the preferred outcome have $\mu_i = -1$. Formally:
\[
U_i^{\sf no} \sim \mathcal{N}(\mu_i,\sigma^2), \quad \mu_i = \begin{cases} +1 & \text{if $c_i = \textsf{winner}$},\\ -1 & \text{otherwise}, \end{cases} \qquad \sigma=1.
\]

This models the intuition that voters who voted against the adversary's preference (which we always model as the losing side) likely have higher intrinsic utility for that outcome.

\paragraph{Tally algorithm-dependent bribe margins.}

For each tally algorithm $\disclosure$, we compute per-voter bribe margins $\bm_i$ as follows:

\begin{itemize}
\item \textbf{Full-disclosure ($\disclosure_{\mathsf{public}}$):} $\bm_i = 1$, using the exact bribe margin from Corollary~\ref{crly:std_cases}
\item \textbf{Winner-only ($\disclosure_{\mathsf{winner}}$):} $\bm_i = \Delta_i$, using the exact bribe margin from Corollary~\ref{crly:std_cases}
\item \textbf{Corrected noised tally ($\disclosure_{\mathsf{noised}(\nu)+}$):}
  \[
  \bm_i = \Delta_i + (1-\Delta_i) \cdot \mathrm{TV}_i,
  \]
  using the upper-bound on bribe margin from Theorem~\ref{thm:noised-bribe-margin}. $\mathrm{TV}_i = 1 - e^{-\beta w_i/2}$ is the value of the total variation distance integral, with $\beta$ being the noise parameter for Laplace($1/\beta$) noise. 
\end{itemize}

Note the for corrected noised tally setting the bribe margin to an upper bound means we compute a lower bound on B-privacy by Theorem~\ref{thm:optimal-bcf}.

\paragraph{Fixed-point computation of equilibrium.}

Given a bribe vector $\mathbf{b}$ and bribe margins $\boldsymbol{\bm}=(\bm_i)$, Theorem~\ref{thm:adv_success_prob} implies that in equilibrium, the probability that voter $i$ votes for the adversary's preferred outcome is:
\[
p_i = \Phi\left(\frac{\bm_i b_i/\Delta_i - \mu_i}{\sigma}\right),
\]
where $\Phi$ is the standard normal CDF and $\Delta_i$ is voter $i$'s pivotality.

The pivotality vector $\boldsymbol{\Delta}=(\Delta_i)$ must satisfy the equilibrium condition: when all other voters play according to the probabilities $(p_j)_{j \neq i}$, voter $i$'s pivotality is:
\[
F_i(\boldsymbol{\Delta}) = \Pr\left[ \sum_{j\neq i} w_j X_j \in [W/2 - w_i, W/2) \right],  X_j \sim \mathrm{Bernoulli}(p_j).
\]

We compute $\boldsymbol{\Delta}$ as the fixed point $\boldsymbol{\Delta} = F(\boldsymbol{\Delta})$, which corresponds to the Bayesian Nash equilibrium condition that no voter wants to deviate given others' equilibrium behavior.

To evaluate the distribution $\sum_{j\neq i} w_j X_j$, we use Monte Carlo with common random numbers and antithetic variates to reduce variance. While costlier than Gaussian approximations, this avoids central limit theorem regularity requirements and yields stable accuracy even under extreme weight disparity.

We iterate $\boldsymbol{\Delta}^{(t+1)}=\cmepsilon F(\boldsymbol{\Delta}^{(t)})+(1-\cmepsilon)\boldsymbol{\Delta}^{(t)}$ with under-relaxation $\cmepsilon
=0.7$ until convergence.

\begin{figure*}[th!]
    \centering
    \includegraphics[width=\textwidth]{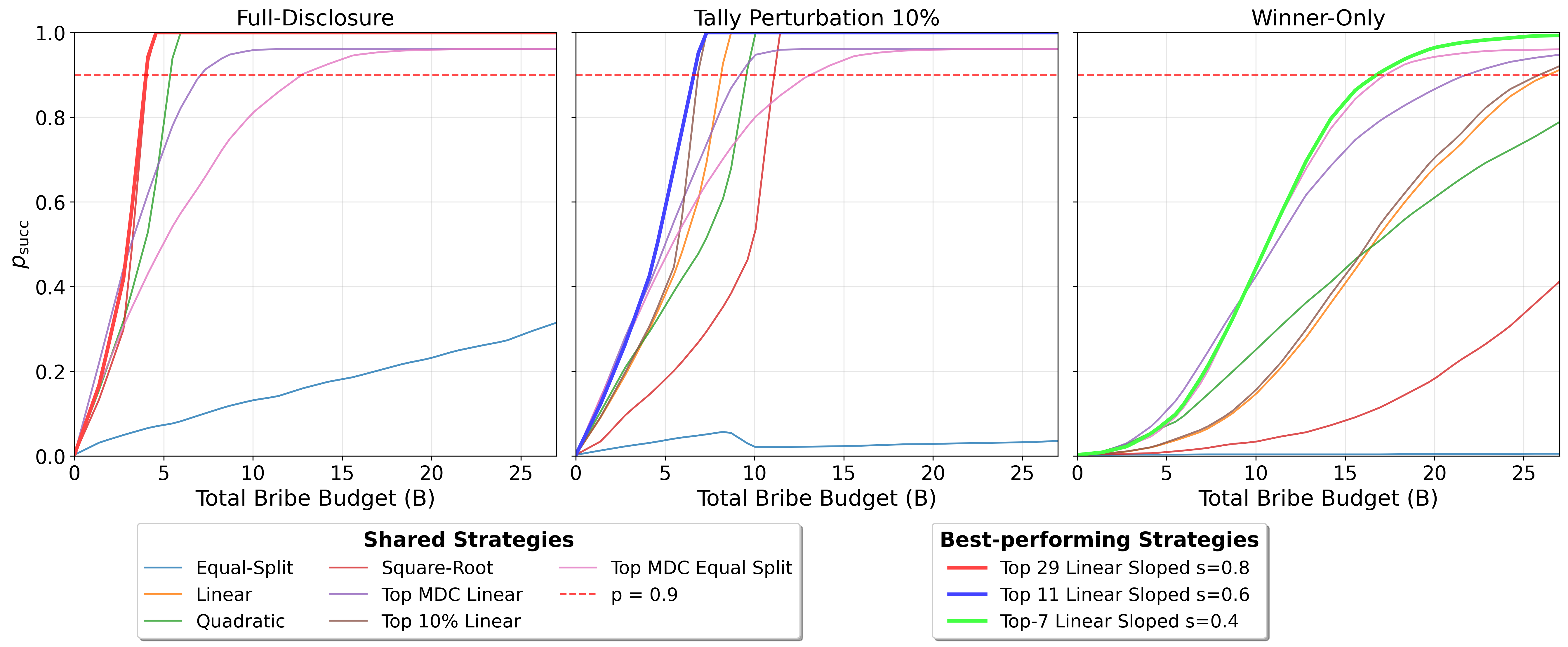}
    \caption{Adversarial success at equilibrium $p_{\mathsf{succ}}$ plotted against total bribe budget B for a range of bribe allocation strategies and tally algorithms for the same ApeCoin proposal from Figure~\ref{fig:Bribe_Distribution}. As expected as the tally algorithm releases less information all strategies are require a higher budget to achieve the same success (shift towards the right on each plot). The target success probability $p=0.9$ used in our experiments is represented by the dashed red line. For each tally algorithm we performed an exhaustive search over allocation strategies and parameters to identify those which required the minimum budget to achieve the target success probability, these curves are bold. Each of these best-performing strategies exhibits similar performance to at least one generic strategy from Table~\ref{tab:main_strategies}.} 
    \label{fig:budget_success_curves}
\end{figure*}

\paragraph{Heuristic bribe allocation strategies.} 
Given $(\mathbf{w}, \boldsymbol{\Delta}, \boldsymbol{\alpha}, B)$ we seek to allocate budget $B$ across voters to maximize the adversary's success probability. We initially attempted standard optimization methods but found that the non-convex nature of the problem caused these approaches to often fail or get trapped in local minima. Instead, we turned to exploration of heuristic bribe allocation strategies that focus on voters unlikely to support the adversary's preferred outcome (since bribing supporters would be wasteful).

Beyond choosing how to weight bribes among targeted voters, we also vary the number of voters to target. Let $k$ denote the number of opposing voters (ranked by weight) that receive bribes, with the remaining voters receiving no bribes. We conducted preliminary testing across a range of allocation strategies applied to varying values of $k$, ranging from the Minimum Decisive Coalition (MDC) size to the total number of opposing voters. Table~\ref{tab:experimental_alloc_func} summarizes the types of weighting strategies we explored across different values of $k$. Based on performance across a subset of proposals, we limited the strategies tested for our main experiments to those enumerated in Table~\ref{tab:main_strategies}.

\begin{table}[h]
\centering
\begin{tabular}{ll}
\toprule
\textbf{Allocation Strategy} & \textbf{Voters Targeted} \\
\midrule
\multirow{3}{*}{Linear} & All voters \\
 & 10 largest voters \\
 & Top 10\% of voters \\
 & Top 1\% of voters \\
\midrule
\multirow{2}{*}{Logarithmic} & Top MDC voters \\
 & Top 1\% of voters \\
\midrule
Square-Root & All voters \\
\midrule
Equal Split & Top MDC voters \\
\bottomrule
\end{tabular}
\caption{Allocation strategies used in main experiments. MDC (Minimum Decisive Coalition) is the smallest number of voters who could flip the outcome.}
\label{tab:main_strategies}
\end{table}


For each allocation strategy, we compute the resulting equilibrium and success probability, then select the strategy that maximizes the adversary's success probability for that budget level. Figure~\ref{fig:budget_success_curves} illustrates how the adversary's success probability varies as budget increases across different allocation strategies under full-disclosure, 10\% tally perturbation, and winner-only scenarios for a representative ApeCoin proposal (the same proposal used in Figure~\ref{fig:Bribe_Distribution}). For this figure, we conducted a more exhaustive enumeration over allocation strategies and values of $k$, and also include the best-performing strategy from this exhaustive search for each tally algorithm, demonstrating that its performance is similar to strategies used in the main experiments.


\paragraph{Success probability estimation.}

Given $(\mathbf{w},\boldsymbol{\Delta},\boldsymbol{\bm},\mathbf{b})$, we compute per-voter success probabilities via the equilibrium condition above and estimate $p_{{\sf succ}}$ by Monte Carlo with variance reduction techniques (common random numbers and antithetic variates). We use $R=1{,}000$ samples by default, increasing to $10{,}000$ if the estimate is within $0.01$ of the target probability $p$.

\paragraph{Summary of B-privacy computation.}
We compute $B^*$ by testing each allocation strategy listed in Table~\ref{tab:main_strategies} to find its minimum required budget, then selecting the best strategy. For each allocation strategy, we use binary search over budget levels to find the minimum budget $B$ that achieves the target success probability $p$. For a given budget level, we compute the resulting success probability via the following process:

\begin{enumerate}
\item Allocate bribes according to the strategy, focusing on voters opposing the adversary's preference
\item Compute the resulting equilibrium by iterating until pivotality convergence:
\begin{enumerate}
\item Recompute bribe margins $\boldsymbol{\alpha}$ if they depend on current pivotalities $\boldsymbol{\Delta}$
\item Update pivotalities based on voter choice probabilities
\end{enumerate}
\item Evaluate the resulting success probability $p_{\mathsf{succ}}$ via Monte Carlo
\end{enumerate}

We approximate the true optimal B-privacy as the minimum budget among all tested allocation strategies.

\end{document}